\documentclass[10pt,aps,pra,twocolumn,floatfix,nofootinbib,superscriptaddress,longbibliography]{revtex4-2}
\usepackage{bm,graphicx,mathrsfs,amsmath,amssymb,mathtools,makecell,bbm,amsthm,dsfont,times,txfonts,nicefrac,framed,enumitem,tikz,physics,lipsum,nicematrix,halloweenmath,tabularray,wrapfig}
\usepackage[most]{tcolorbox}
\usepackage{pifont}
\usepackage{etoolbox}
\usepackage{xargs}

\usepackage[colorlinks=true,linkcolor=equationcolor,citecolor=refcolor]{hyperref}
\hypersetup{urlcolor=refcolor}

\definecolor{equationcolor}{RGB}{222,94,100}
\definecolor{refcolor}{RGB}{141,138,249}
\definecolor{changescolor}{RGB}{188,104,104}
\definecolor{alecolor}{RGB}{198,113,190}
\newenvironment{sproof}{%
  \proof}{\endproof}

\makeatletter
\def\blfootnote{\gdef\@thefnmark{}\@footnotetext}
\makeatother

\renewcommand{\S}{\ms{S}}

\renewcommand{\v}[1]{\ensuremath{\boldsymbol #1}}
\newcommand{\ms}[1]{\textsf{#1}}

\newcommand{\iden}{\mathbbm{1}}
\newcommand{\stab}{\mathrm{STAB}_1}

\def\X{ {\ms X} }

\def\E{ {\ms E} }
\def\M{ {\ms M} }

\def\P{ {\cal P} }

\newcounter{globalboxcounter} 
\renewcommand{\theglobalboxcounter}{\arabic{globalboxcounter}}

\newtcolorbox[use counter=globalboxcounter]{mybox}[2][]{%
    breakable,
    enhanced,
    sharp corners,
    colback=white,
    colframe=black!50,
    colbacktitle=white,
    coltitle=black,
    fonttitle=\normalfont,
    boxrule=0.25mm,
    arc=0mm,
    width=\linewidth,
    boxsep=1mm,
    left=1mm,
    right=1mm,
    top=1mm,
    bottom=1mm,
    before skip=6pt,
    after skip=6pt,
    before upper=\strut,
    title={\centering\strut \textbf{Example} {\normalfont\theglobalboxcounter: #2}},
    #1
}

\newtheorem{thm}{Theorem}

\newtheorem{lem}[thm]{Lemma}
\newtheorem{prop}[thm]{Proposition}
\newtheorem{cor}[thm]{Corollary}

\begin{document}
\title{Every Little Thing Heat Does Is Magic }

\author{Rafael A. Macêdo}
    \affiliation{Physics Department, Federal University of Rio Grande do Norte, Natal, 59072-970, Rio Grande do Norte, Brazil}
    \affiliation{International Institute of Physics, Federal University of Rio Grande do Norte, 59078-970, Natal, RN, Brazil}
\author{A. de Oliveira Junior}
	\affiliation{Center for Macroscopic Quantum States bigQ, Department of Physics,
Technical University of Denmark, Fysikvej 307, 2800 Kgs. Lyngby, Denmark}
\author{Naim E. Comar}
    \affiliation{International Institute of Physics, Federal University of Rio Grande do Norte, 59078-970, Natal, RN, Brazil}
\author{Luna Lima Keller}
    \affiliation{International Institute of Physics, Federal University of Rio Grande do Norte, 59078-970, Natal, RN, Brazil}
\author{Jonatan Bohr Brask}
	\affiliation{Center for Macroscopic Quantum States bigQ, Department of Physics,
Technical University of Denmark, Fysikvej 307, 2800 Kgs. Lyngby, Denmark}
\author{Lucas C. C\'eleri}
\affiliation{QPequi Group, Institute of Physics, Federal University of Goi\'as, Goi\^ania, Goi\'as, 74.690-900, Brazil}
\author{Rafael Chaves}
    \affiliation{International Institute of Physics, Federal University of Rio Grande do Norte, 59078-970, Natal, RN, Brazil}
    \affiliation{School of Science and Technology, Federal University of Rio Grande do Norte, Natal, Brazil}
\date{\today}

\begin{abstract}
How can one certify that an unknown quantum state possesses magic without resorting to full state tomography? We address this question by introducing two thermodynamic witnesses that rely solely on energy and heat measurements. First, we define the stabilizer ground‑state energy as the lowest energy achievable by any stabilizer state, and the stabilizer gap as the separation between this value and the true ground‑state energy. Any state whose energy lies below the stabilizer ground‑state energy is therefore necessarily nonstabilizer. This leads to a direct witness of magic using only average-energy measurements. To overcome the limitations when direct energy measurements are inconclusive, we further develop a nonlinear witness based on heat exchange with a thermal ancilla. Specifically, we derive fundamental bounds on heat that are satisfied by all stabilizer states; therefore, their violation certifies the presence of magic. We demonstrate the effectiveness of our approach through several examples, ranging from few-body systems where heat exchange reveals nonstabilizerness even when energy measurements alone fail, to the transverse-field Ising chain, where the stabilizer gap becomes maximal at the quantum critical point.
\end{abstract}

\maketitle

\section{Introduction}

A fundamental problem in quantum science is determining whether a quantum device works as intended. Across many areas of quantum information, we are faced with the task of characterizing states of interest. In principle, one could perform full quantum state tomography~\cite{James2001,Banaszek2013} and obtain a complete description of the state at hand. In practice, such a procedure quickly becomes infeasible as the size of the system increases. A complete description requires a number of parameters that scale exponentially with system size, and so does the experimental effort required for reconstruction. This curse of dimensionality forces us to look for alternative approaches that can certify relevant properties without requiring full reconstruction~\cite{Gross2010,Aaronson2018,Huang2020,Koutn2022}.

Two important instances of this general problem, each plagued by the curse of dimensionality, are entanglement and nonstabilizerness. Entanglement captures nonclassical correlations and is a hallmark of quantum physics~\cite{Horodecki2009}; nonstabilizerness is the holy grail that allows one to speak about universal quantum computation. Together both allow one to distinguish genuine quantum advantage~\cite{gottesman1997stabilizercodesquantumerror,Aharonov1997,gottesman1998,aaronson2004improved}. When it comes to deciding whether a state is entangled or not, entanglement theory ran into this wall of dimensionality a long time ago. Consequently, a plethora of inequivalent criteria and measures are available for the detection and characterization of entanglement~\cite{peres1996separability,Horodecki1996,Wootters98,Wootters2001,bruss2002,Plenio2005}. In particular, entanglement witnesses form a rich and well-developed framework~\cite{terhal2001family,lewenstein2000optimization,guhne2009entanglement}. On the other hand, while the resource theory of nonstabilizerness has a solid theoretical foundation~\cite{Veitch2014, howard2017application,Leone2024}, the certification and detection of magic have only recently attracted significant attention, with a variety of new methods beginning to emerge~\cite{Leone2023,Tirrito2024,wagner2024certifyingnonstabilizernessquantumprocessors,Macedo2025,Warmuz2025,varela2026predictingmagicmeasurements,leone2026unbearablehardnessdecidingmagic,Zamora_2025}.

Among these approaches, one particularly appealing idea is to use the Hamiltonian itself as a witness, reducing the problem to estimating a single, physically meaningful observable. For example, in the case of entanglement, if no separable state can reach the measured energy, then the state must be entangled~\cite{Tth20051,Tth2005,Guhne2005,Lidar2005,Ghne2006}. This leads to the notion of an entanglement gap~\cite{Dowling2004}, which quantifies the energy range that is inaccessible to separable states. This motivates the central question addressed in this work: can nonstabilizerness be witnessed using only energy expectation values of a given Hamiltonian?

We answer this question in the affirmative. Motivated by entanglement theory and thermodynamics, we introduce energy- and heat-based witnesses of nonstabilizerness that certify magic using experimentally accessible observables. More precisely, we introduce the notion of the stabilizer gap, defined as the difference between the energy of the ground-state and the minimum energy achievable by any stabilizer state. If the system's energy lies below the stabilizer threshold, the state is necessarily nonstabilizer. Consequently, witnessing nonstabilizerness requires only direct energy measurements. However, this method is blind to Hamiltonians whose ground spaces contain stabilizer states. To overcome this limitation, we develop a nonlinear witness based on heat exchange between the system and a thermal ancilla. The corresponding bounds are derived in a memory-assisted thermodynamic framework and are therefore respected a fortiori by standard thermal operations without memory. Violations of these bounds certify nonstabilizerness. Operationally, the experimental input is the energetic pair $(E_0,Q)$ that consists of the system's average energy and the measured heat exchange. We illustrate these methods with a range of representative examples. For the energy witness, we analyze the transverse field Ising chain, where the stabilizer gap peaks at the quantum critical point and vanishes in the limits of vanishing or infinite transverse field, showing that the ground state's magic is thermodynamically detectable via its energy alone. For the heat witness, we examine noisy and dephased families of states with fixed average energy--where direct energy measurements are inconclusive--and demonstrate violations of the stabilizer heat bounds beyond the single‑qubit case.

The paper is organized as follows. In Sec.~\ref{S:stabilizer-resource-theory}, we review the stabilizer resource theory, define the stabilizer polytope, and introduce the class of stabilizer Hamiltonians. Section~\ref{S:thermodynamic-framework} presents the thermodynamic framework and derives the optimal heat exchange bounds that will serve as our witnesses. In Sec.~\ref{S:measurement-witnessing}, we introduce the two measurement models--direct energy measurement and indirect heat measurement--and define the stabilizer gap and the stabilizer heat bounds. Section~\ref{S:stabilizer-energy-witness} analyzes the direct energy witness in detail, including the stabilizer gap, its properties, examples, and algorithms for its computation. Section~\ref{S:stabheatwit} focuses on the heat‑based witness, providing a complete geometric characterization for single qubits, optimality conditions, and a three‑qubit example showing detection even when direct energy measurements fail. We conclude with an outlook in Sec.~\ref{S:outlook}. Technical derivations, including the reduction to commuting Pauli subsets and the degenerate perturbation analysis, are given in the appendices.

\section{Stabilizer resource theory}
\label{S:stabilizer-resource-theory}

We start by reviewing the basics of stabilizer resource theory. 

For a system of $n$ qubits with Hilbert space \mbox{$\mathcal H=(\mathbb C^2)^{\otimes n}$}, we denote by $\widetilde{\mathcal P}_n := \{\pm 1,\pm i\}\cdot\{\iden,X,Y,Z\}^{\otimes n}$ the $n$-qubit Pauli group and by $\mathcal P_n := \{P\in \widetilde{\mathcal{P}}_n : P^\dagger= P\}$ the set of all Hermitian elements of the $n$-qubit Pauli group. Since all Hamiltonians considered here are Hermitian, Pauli expansions will use elements of $\mathcal P_n$. 

A stabilizer group $\mathbf S$ is an Abelian subgroup of $\widetilde{\mathcal P}_n$ that does not contain the element $-\iden$. This group defines a stabilizer code as the subspace $V_\mathbf S \equiv \{\ket{\psi} \in \mathcal H\;|\;s \ket{\psi} =+\ket{\psi}, \;\forall s \in \mathbf S\}$, with dimension $\dim V_\mathbf S =2^{n-\rank(\mathbf S)}$. Here, $\rank(\mathbf S)$ denotes the number of independent generators of $\mathbf S$. When $\rank(\mathbf S)=n$, the code is one dimensional, and we say that $\mathbf S$ stabilizes a unique stabilizer state, denoted $\ket{\mathbf S} \in \mathcal H $. Any stabilizer group can be compactly described by a set of generators $\mathrm{gen}(\mathbf S)$. The group $\mathbf{S}$ itself is then the set of all products of these generators, written as $\langle \mathrm{gen}(\mathbf{S}) \rangle$.

The set of rank-$n$ stabilizer groups, referred to as maximal stabilizer groups, is in one-to-one correspondence with the set of pure stabilizer states. Let $\operatorname{ext}(\textrm{STAB}_n)$ denote the set of all such pure stabilizer states $\ket{\mathbf S}$. The stabilizer polytope is then defined as the convex hull of these extremal points:
\begin{align}
    \mathrm{STAB}_n &\equiv \mathrm{conv} \{\ketbra{\mathbf S}\}_{\mathbf S \in \mathrm{ext}(\mathrm{STAB}_n)} \nonumber\\
    &= \qty{\sum_{\mathbf S \in \mathrm{ext}(\mathrm{STAB}_n)}p_\mathbf S \ketbra{\mathbf S} : p_\mathbf S \geq 0, \; \sum_\mathbf S p_\mathbf S=1}\;.
\end{align}
A quantum state $\rho$ is a nonstabilizer state if $\rho \notin \mathrm{STAB}_n$. The number of extremal points, which coincides with the number of pure stabilizer states of $n$ qubits, is~\cite{aaronson2004improved}
\begin{equation}
    |\mathrm{ext}(\mathrm{STAB}_n)| = 2^n \prod_{k=1}^n (2^k+1) = 2^{O(n^2)}.
\end{equation}
Hence, establishing membership in the polytope for a sufficiently large number of qubits is difficult. Full state tomography requires data that scale exponentially with system size, and a naive vertex-based membership test or geometric optimization involves a search over $2^{O(n^2)}$ extremal stabilizer states. This is a common feature across mixed-state quantum resource theories~\cite{chitambar2019quantum}, since many resource measures reduce to geometric quantities in which one searches for an optimal convex decomposition of a quantum state in terms of the free set.

There is an alternative description of the stabilizer polytope in terms of its facets. More precisely, it admits the following ($H$-)representation:
\begin{equation}
    \mathrm{STAB}_n = \{\rho \in \mathcal S(\mathcal H)\;|\;\tr(A\rho) \geq 0,\; \forall A \in \mathrm{F}_n\}\;,
    \label{eq:Hrepstab}
\end{equation}
where $S(\mathcal H)$ denotes the set of density operators in $\mathcal{H}$ and $\mathrm{F}_n$ is a finite set of facet-defining Hermitian operators. These operators correspond to the vertices of the $\Lambda$-polytope, which is dual to $\mathrm{STAB}_n$~\cite{zurel2020hidden}. Although it is known that obtaining the facets ($H$-representation) from the vertex description ($V$-representation) of a polytope is NP-hard, this representation shows that, in principle, there exists a finite set of witnesses that characterizes the polytope, and hence these operators are natural candidates for detecting the resource.

\begin{mybox}{One- and two-qubit stabilizer polytopes \& facets}
For a single qubit $\rho(\v{r})=\tfrac12(\iden+\v{r}\cdot\v{\sigma})$, where $\v r= (r_x, r_y, r_z)$ is the Bloch vector with $\|\v r\|_2 \leq 1$ and $\v \sigma = (X,Y,Z)$, the stabilizer polytope is the octahedron
\begin{equation}
    \mathrm{STAB}_1=\mathrm{conv}\{\rho(\pm\v e_x),\rho(\pm\v e_y),\rho(\pm\v e_z)\},
\end{equation}
where $\v e_w$ are the Cartesian basis vectors in Bloch space. Equivalently, these six vertices are the eigenstates of the Pauli operators $X$, $Y$, and $Z$. A convenient set of facet operators is
\begin{equation}
    \mathrm{F}_1=\{\iden-\v f\cdot\v{\sigma} \mid \v f\in\{\pm1\}^3\},
\end{equation}
since $\tr[(\iden-\v f\cdot\v{\sigma})\rho(\v{r})]=1-\v f\cdot\v{r}\ge0$ $\forall\, \v f\in\{\pm1\}^3$ is equivalent to the octahedral condition $\|\v r\|_1\le1$.

\vspace{0.2cm}
For two qubits the stabilizer polytope has $|\mathrm{ext}(\mathrm{STAB}_2)|=60$ pure stabilizer states and $|\mathrm{F}_2|=22,\!230$ facets~\cite{zurel2020hidden,Palhares2026}. The high polyhedral symmetry of the one‑qubit case is lost. 
\end{mybox}

Since energy measurements play a central role below, it is useful to single out a special class of Hamiltonians within stabilizer resource theory. Given a stabilizer group $\mathbf S\subseteq\widetilde{\mathcal P}_n$ and a chosen generating set $\mathrm{gen}(\mathbf S)$, we define a stabilizer Hamiltonian associated with $\mathbf S$ by
\begin{equation}
H[\mathbf S,\mathrm{gen}(\mathbf S)] := -\sum_{g\in\mathrm{gen}(\mathbf S)} g.
\label{Eq:stabilizer-Hamiltonian}
\end{equation}
Notice that $H[\mathbf S,\mathrm{gen}(\mathbf S)]$ depends on the chosen independent generating set. Different choices of generators for the same stabilizer group $\mathbf S$ may lead to different Hamiltonians, although they all share the same ground space $V_{\mathbf S}$. More generally, we say that $H\in \operatorname{Hem}(\mathcal{H})$ is a stabilizer Hamiltonian if there exists a stabilizer group $\mathbf S$ and an independent generating set $G=\operatorname{gen}(\mathbf S)$, such that $H=H[\mathbf S,G]$.

\section{Thermodynamic framework \label{S:thermodynamic-framework}}

In this section, we introduce the thermodynamic framework used throughout the paper. We first recall the resource-theoretic description based on thermal operations and then extend it to a memory-assisted scenario in which a quantum memory allows one to exploit energy coherences to enhance heat exchange~(see~\cite{Gour2015,Lostaglio2019,deOliveiraJunior2022} for a detailed discussion about resource theory of thermodynamics). This leads to fundamental bounds on the heat that can be exchanged in cooling and heating processes, which we will later connect to the presence of nonstabilizerness.

\subsection{Thermodynamic setting} \label{S:thermodynamic-setting}
We consider a system $\ms S$, described by a Hamiltonian $H_{\ms S}$ and an unknown state $\rho_{\ms S}$, along with a thermal environment $\ms{E}$ with Hamiltonian $H_{\ms E}$ at inverse temperature $\beta$. The environment is in a thermal Gibbs state,
\begin{equation}
    \gamma_{\ms E} = \frac{e^{-\beta H_{\ms E}}}{\tr(e^{-\beta H_{\ms E}})}.
\end{equation}
The joint system is assumed to be closed and to evolve unitarily, with the only constraint being that the overall energy is conserved. Hence, the set of thermodynamic transformations is modeled by thermal operations~\cite{Janzing2000,horodecki2013fundamental}, defined as completely positive trace-preserving maps $\mathcal{E}$ that can be realized by coupling $\ms S$ to a thermal environment $\ms E$ via an energy-preserving unitary. Formally, for any input state $\rho_{\ms S}$ of the system,
\begin{equation}
    \mathcal{E}(\rho_{\ms S})  := \tr_{\ms E} \qty[U\qty(\rho_{\ms S} \otimes \gamma_{\ms E})U^{\dagger}],
    \label{Eq:thermal-operation}
\end{equation}
where $U$ is a joint energy-preserving unitary on $\ms S \otimes \ms E$ satisfying $[U, H_{\ms S} \otimes \iden_{\ms E} + \iden_{\ms S} \otimes H_{\ms E}]=0$. The joint system-bath state after the interaction $U$ is denoted by $\eta_{\ms{SE}}:=U(\rho_{\ms S}\otimes \gamma_{\ms E})U^{\dagger}$. Thus, the corresponding marginal states are $\eta_{\ms S} := \tr_{\ms E}(\eta_{\ms{SE}})$ and $\eta_{\ms E} := \tr_{\ms S}(\eta_{\ms{SE}})$.

Since the global unitary $U$ conserves the total energy, the energy lost by $\ms S$ equals the energy gained by $\ms E$. Defining the average energy change of the system and environment as $\Delta E_{\ms S} := \tr[H_{\ms S}(\eta_{\ms S} - \rho_{\ms S})]$ and $\Delta E_{\ms E} := \tr[H_{\ms E}(\eta_{\ms E} - \gamma_{\ms E})]$, energy conservation implies $\Delta E_{\ms S} + \Delta E_{\ms E} = 0$. We regard the energy gained by the environment as heat and define
 \begin{equation}
     Q(\rho_{\ms S}) := \Delta E_{\ms E} = \tr[H_{\ms E}(\eta_{\ms E} - \gamma_{\ms E})]. \label{Eq:heat-env-def}
 \end{equation}
Due to energy conservation, this is equal to minus the energy change of the system $Q(\rho_{\ms S}) = -\Delta E_{\ms S}$. Thus, a positive value of $Q(\rho_{\ms S})$ corresponds to the heat flowing from $\ms S$ into the environment. In the following, we will focus on $Q(\rho_{\ms S})$ and view the environment $\ms E$ as a thermal ancilla whose energy change we can measure. 

A more specific structural constraint comes from energy conservation: since $[U, H_{\ms S}+H_{\ms E}]=0$, the unitary $U$ is block-diagonal in the eigenspaces of the total Hamiltonian. As a consequence, coherences between different total-energy blocks do not contribute to changes of local average energies under energy-preserving dynamics (see, e.g., Appendix B of Ref.~\cite{lipkabartosik2023fundamental}). In the standard setting of thermal operations, the bath is initially energy-diagonal, so any coherence of $\rho_{\ms S}$ between distinct eigenenergies contributes only to off-block terms of $\rho_{\ms S}\otimes\gamma_{\ms E}$ and therefore cannot affect $\Delta E_{\ms S}$. Hence, for thermal operations, the average heat exchange depends only on the energy populations of $\rho_{\ms S}$.

This limitation can be overcome by extending the setting of thermal operations to include a quantum memory $\ms M$. We adopt the memory-assisted framework of Ref.~\cite{de2025heat}, where $\ms M$ allows one to access otherwise thermodynamically inaccessible coherences without supplying net energy. Formally, we introduce an additional system $\ms M$ with Hamiltonian $H_{\ms M}$, initially prepared in an arbitrary state $\rho_{\ms M}$ and uncorrelated with $\ms S$ and $\ms E$. We then consider thermal operations on the composite system $\ms S \otimes \ms M \otimes \ms E$, with $U$ acting jointly on all three subsystems and satisfying the original energy-conservation condition, now extended to include the memory. We impose one additional crucial constraint: The memory must be returned to its initial state, $\eta_{\ms M} := \tr_{\ms{SE}}[U(\rho_{\ms S}\otimes\rho_{\ms M}\otimes\gamma_{\ms E})U^{\dagger}] = \rho_{\ms M}$. In other words, we allow the final global state $\eta_{\ms{SME}}$ to be correlated, and the only requirement is that the memory is returned locally unchanged. Since $\eta_{\ms M}=\rho_{\ms M}$, the average energy of the memory is unchanged, $\Delta E_{\ms M}=0$. Together with $[U,H_{\ms S}+H_{\ms M}+H_{\ms E}]=0$, this implies $\Delta E_{\ms E}=-\Delta E_{\ms S}$ and, hence, $Q(\rho_{\ms S})=\tr[H_{\ms S}(\rho_{\ms S}-\eta_{\ms S})]$. Therefore, the memory does not become an extra source or sink of energy.

In this memory-assisted regime, energy coherences that were thermodynamically silent can now be exploited. The quantum memory, which may contain phase information or correlations with respect to its own Hamiltonian, provides a reference that allows the global unitary $U$ to convert system coherences into population changes in $\ms S$ and $\ms E$. Specifically, because the memory is returned unchanged, interactions can effectively transfer the informational content of the coherences into population changes, changing the heat exchange. As a result, the achievable heat exchange $Q(\rho_{\ms S})$ can exceed the limits imposed by purely classical or incoherent protocols~\cite{lipkabartosik2023fundamental, Czartowski2024,de2025heat}. One can then ask: what are the fundamental limits on the heat that can be absorbed or released by the environment when such a quantum memory is available?

\subsection{Optimal heat exchange}

To address this question, we quantify the enhancement by defining optimal heat exchange as the extremal value of $Q(\rho_{\ms S})$ achievable under the memory-assisted thermal operations framework. Following Ref.~\cite{de2025heat}, we consider the optimal heat that can be exchanged with the environment in cooling and heating processes, under the constraints of energy-preserving dynamics and a cyclic memory. We define the minimal and maximal heat achievable for a given input $\rho_{\ms S}$ as
\begin{align}
\label{Eq:app-optimal-heat}
\begin{split}
    Q_{c/h}(\rho_{\ms S}) :=   \underset{H_{\ms E},\, H_{\ms M}, \, U, \, \rho_{\ms{M}}}{\min/\max} \:\: &\tr[H_{\ms{E}}(\eta_{\ms{E}} - \gamma_{\ms{E}})], \\ 
\textrm{s.t.}  \quad\quad &[U, H_{\ms{S}} +H_{\ms{M}} + H_{\ms{E}}] = 0, \\   &\eta_{\ms M}=\rho_{\ms{M}}. 
\end{split}
\end{align}
Note that optimization ranges over all environment and memory Hamiltonians, all memory states $\rho_{\ms{M}}$, and all energy-preserving unitaries $U$. The minimum (maximum) corresponds to optimal cooling (heating) of the environment. These quantities $Q_{c}(\rho_{\ms S})$ and $Q_{h}(\rho_{\ms S})$ represent fundamental bounds on heat exchange within the memory-assisted thermal operations framework: no protocol satisfying energy conservation and the memory constraint can cool or heat the environment beyond these limits. With our sign convention, a positive $Q$ denotes heat deposited in the environment. Therefore, cooling the environment corresponds to $Q<0$, i.e., minimizing $Q$, while heating corresponds to maximizing $Q$.

At first sight, the optimization in Eq.~\eqref{Eq:app-optimal-heat} looks untractable. Remarkably, using techniques from the resource theory of thermodynamics and catalytic thermal operations~\cite{Brandao2013,LipkaBartosik2024,Datta2023}, it can be reduced to a much simpler problem and solved in a closed form. Let $F_\beta(\rho) := \tr[H_{\ms S} \rho] - \beta^{-1} S(\rho)$ denote the nonequilibrium free energy of a state $\rho$ at inverse temperature $\beta$, with $S(\rho) := -\tr(\rho \log \rho)$ the von Neumann entropy, and let $\gamma_{\ms S}(x)$ be the Gibbs state of $\ms S$ at inverse temperature $x$. One can show that the optimal final states are Gibbs-form states $\gamma_\S(\beta_{c/h})$. Defining the average energy $E(\rho) := \tr(H_{\ms S}\rho)$, one finds the following compact expression for the solution of Eq.~\eqref{Eq:app-optimal-heat}:
\begin{equation}
    Q_{c/h}(\rho_{\ms S})
    = E(\rho_{\ms S}) - E[\gamma_{\ms S}(\beta_{c/h})].
    \label{Eq:QcQh-solution}
\end{equation}
Intuitively, the optimal protocol maps the system to a thermal state $\gamma_{\ms S}(\beta_{c/h})$, and the optimal heat is just the difference between the initial and final average energy of the system. The nontrivial part of the optimization is encoded in the effective temperatures $\beta_{c/h}$.

These effective inverse temperatures are determined by the condition that the free energy of the final Gibbs state equals the free energy of the initial state, i.e. $F_\beta[\gamma_{\ms S}(\beta_{c/h})] = F_\beta(\rho_{\ms S})$. For finite-dimensional systems, the function $x \mapsto F_\beta[\gamma_{\ms S}(x)]$ is convex. Therefore, for a given $\rho_{\ms S}$, equation $F_\beta[\gamma_{\ms S}(x)] = F_\beta(\rho_{\ms S})$ can have zero, one, or two solutions. The cooling optimum $\beta_c$ is the lower inverse temperature solution (if it exists), and the heating optimum $\beta_h$ is the higher. In particular, in some parameter regimes, the equation has no solution for the heating branch, and then the optimum is attained at the energetic extreme (see Appendix A of~\cite{de2025heat}). In equilibrium [$\rho_{\ms S}=\gamma_{\ms S}(\beta)$], the two solutions coincide with the inverse temperature of the bath $\beta$. The corresponding heat is given by Eq.~\eqref{Eq:QcQh-solution}.

The optimal values $Q_{c/h}(\rho_{\ms S})$ depend on the system's nonequilibrium free energy $F_\beta(\rho_{\ms S})$ and thus on both its energy and its entropy. This dependence can be turned into a tool for witnessing quantum resources when we constrain the initial state $\rho_{\ms S}$ to a subset $\mathcal{S}$ of operationally relevant states (e.g., separable, incoherent, or stabilizer states). For that, just define the set-dependent bounds $Q_c(\mathcal{S}):=\min_{\rho\in\mathcal{S}}Q_c(\rho)$ and $Q_h(\mathcal{S}):=\max_{\rho\in\mathcal{S}}Q_h(\rho)$. By construction, any state $\rho\in\mathcal{S}$ gives a heat exchange $Q$ satisfying $Q_c(\mathcal{S}) \leq Q \leq Q_h(\mathcal{S})$. The contrapositive provides the witnessing power: if an experimental measurement of heat exchange yields a value $Q$ such that $Q < Q_c(\mathcal{S})$ or $Q > Q_h(\mathcal{S})$, it certifies that the initial state $\rho_{\ms S}$ possessed resources beyond those available in $\mathcal{S}$---for instance, nonstabilizerness when $\mathcal{S}$ is the set of stabilizer states. Importantly, computing $Q_{c/h}(\mathcal{S})$ for a specific set $\mathcal{S}$ is a separate optimization problem that depends on the structure of $\mathcal{S}$ and the Hamiltonian of the system $H_{\ms S}$.

\section{Measurement-based witnessing of nonstabilizerness} \label{S:measurement-witnessing}

We now address the central operational question of this work: \emph{how can an experimenter certify that an unknown state $\rho_{\ms S}$ lies outside the stabilizer polytope when only limited (energetic) measurements are available?}

In principle, one could perform full quantum state tomography and solve the membership problem $\rho_\S \in \mathrm{STAB}_n$. However, as discussed in Section~\ref{S:stabilizer-resource-theory}, this approach is infeasible for large systems due to the exponential cost of tomography and the super-exponential size of the stabilizer polytope. We therefore focus on witnessing nonstabilizerness using observables that arise naturally in thermodynamic experiments.
\begin{figure*}
    \centering
    \includegraphics{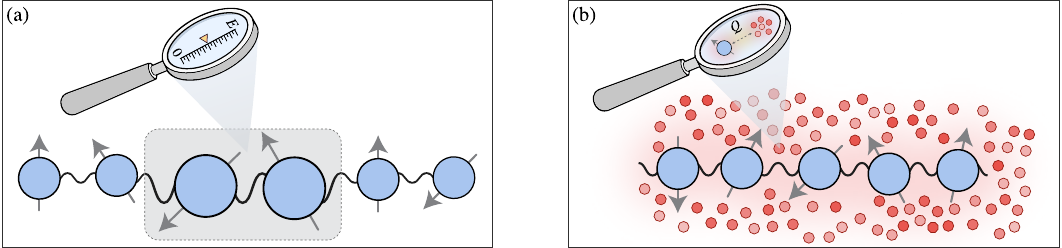}
        \caption{\textbf{Measurement models for witnessing nonstabilizerness.} (a) \emph{Direct energy measurement}: The observer measures the system's energy directly via the Hamiltonian $H_S$. (b) \emph{Indirect heat measurement}: The system interacts with a thermal environment through an energy-conserving unitary $U$, with assistance from a quantum memory. The observer measures only the heat $Q$ deposited in the environment. In both cases, stabilizer states obey bounds ($E_\mathrm{STAB}$ and $Q_\mathrm{STAB}$, respectively) that can be exceeded only by nonstabilizer states.}
    \label{F-measurement-models}
\end{figure*}
We consider a $n$-qubit system $\ms S$ with Hamiltonian $H_{\ms S}$. The experimenter has access to one of the following measurement models (Fig.~\ref{F-measurement-models}):
\begin{enumerate}
    \item [(a)] \emph{Direct energy measurement}: estimate the average energy $E(\rho_{\ms S}) := \tr(H_{\ms S}\rho_{\ms S})$.
    \item [(b)] \emph{Indirect heat measurement}: couple $\ms S$ to a thermal environment and estimate the average heat $Q$ exchanged with the thermal ancilla. 
\end{enumerate}
When a large number of qubits are involved, not all Hamiltonians can be measured. A natural restriction is to consider what are called $k$-local Hamiltonians, specified by a Pauli expansion of the form:
\begin{equation}\label{Eq:Hamiltonian-pauli-decomposition}
    H_\ms S = -\sum_{P \in \mathcal{P}_n} w_P P\;,
\end{equation}
 
Since $H_\ms S$ must be Hermitian and $\mathcal{P}_n$ already contains both $P$ and $-P$, $w_P$ must be real numbers and we may absorb the sign of each coefficient in the corresponding Pauli string; define $\mathbf P(H_{\ms S}):=\{P\in \mathcal{P}_n:w_P\neq 0\}\subseteq \mathcal P_n$, and let us assume without loss of generality that $w_P>0$ for all $P\in \mathbf P(H_{\ms S})$. We say $H_\ms S$ is $k$-local if two constraints are imposed:
\begin{enumerate}
    \item \emph{Sparsity:} $|\mathbf P(H_{\ms S})|=\mathrm{poly}(n)$, i.e., only $\mathrm{poly}(n)$ coefficients $w_P$ are non-zero.
    \item \emph{Locality:} $\max_{P\in \mathbf P(H_{\ms S})}|\mathrm{supp}(P)|=k=O(1)$, where $\mathrm{supp}(P)$ is the set of qubits on which $P$ acts nontrivially. That is, every Pauli string $P$ with $w_P \neq 0$ acts non-trivially on at most $k=O(1)$ qubits. 
\end{enumerate}
Under these assumptions, $E(\rho_{\ms S})$ can be estimated by measuring the local terms $\{P\}_{P\in\mathbf P(H_{\ms S})}$ and combining the results with the known coefficients $\{w_P\}$, thus avoiding the curse of dimensionality in the data required by the full tomography approach.

\subsection{Direct energy measurement}\label{SS:direct-measurement}

In the direct model [see Fig.~\hyperref[F-measurement-models]{\ref{F-measurement-models}(a)}] the experimenter estimates $\tr(H_{\ms S}\rho_{\ms S})$. A standard construction in resource theories is to compare a measured expectation value with its extremal value in the free set~\cite{chitambar2019quantum}. We define the stabilizer ground state energy as
\begin{equation}\label{Eq:stabilizer-ground-energy}
    E_\mathrm{STAB}(H_S) \equiv \min_{\rho \in \mathrm{STAB}_n} \tr(H_S \rho).
\end{equation}
Because $\mathrm{STAB}_n$ is convex, the minimum is attained in an extremal stabilizer state. The Hermitian operator
\begin{equation}
   W_H := H_{\ms S}-E_{\mathrm{STAB}}(H_{\ms S}) \iden,
\end{equation}
is therefore a linear witness of nonstabilizerness. For all stabilizer states $\sigma\in\mathrm{STAB}_n$ one has $\tr(W_H\sigma)\ge 0$, while any state $\rho_{\ms S}$ satisfying
\begin{equation}
    \tr(W_H\rho_{\ms S})<0 \quad \Longleftrightarrow \quad \tr(H_{\ms S}\rho_{\ms S})<E_{\mathrm{STAB}}(H_{\ms S})
    \label{eq:witness-stab}
\end{equation}
is certified to be nonstabilizer.

The usefulness of this witness is quantified by the \emph{stabilizer gap}, 
\begin{equation}\label{Eq:stabilizer-gap}     
    \delta_{\mathrm{STAB}}(H_{\ms S}) \equiv E_{\mathrm{STAB}}(H_{\ms S})-E_{\mathrm{gs}}(H_{\ms S}), 
\end{equation} 
where $E_{\mathrm{gs}}(H_{\ms S}) \equiv \min_{\rho\in\mathcal{S}(\mathcal{H})}\tr(H_{\ms S}\rho)$ is the true ground-state energy of $H_{\ms S}$. The quantity $\delta_{\mathrm{STAB}}(H_{\ms S})$ measures the extent to which the lowest-energy sector of $H_{\ms S}$ is inaccessible to the stabilizer states. In particular, when $\delta_{\mathrm{STAB}}(H_{\ms S})>0$, a sufficiently low energy certifies nonstabilizerness. The following proposition collects the basic properties of this quantity:
\begin{prop}[Stabilizer gap]
    For any $H\in\mathrm{Herm}(\mathcal H)$
\begin{enumerate}
    \item $\delta_{\mathrm{STAB}}(H) \ge 0$;
    \item $\delta_{\mathrm{STAB}}(H)=0$ if and only if the ground eigenspace of $H$ contains a stabilizer state;
    \item if $H$ is a stabilizer Hamiltonian [Eq.~\eqref{Eq:stabilizer-Hamiltonian}], then $\delta_{\mathrm{STAB}}(H)=0$.
\end{enumerate}
\end{prop}
\begin{proof}
(1) holds because $\mathrm{STAB}_n\subseteq\mathcal S(\mathcal H)$, so the minimum over all states cannot exceed the minimum over stabilizer states. For (2), if $\delta_{\mathrm{STAB}}(H)=0$ then $E_{\mathrm{STAB}}(H)=E_{\mathrm{gs}}(H)$, and by compactness some stabilizer state achieves this value, hence it lies in the ground eigenspace. Conversely, if a stabilizer state belongs to the ground eigenspace, it achieves $E_{\mathrm{gs}}(H)$, so $E_{\mathrm{STAB}}(H)=E_{\mathrm{gs}}(H)$. (3) follows because the ground space of a stabilizer Hamiltonian is a stabilizer code that contains stabilizer states.
\end{proof}

Operationally, direct energy measurement is appealing because it relies only on estimating the expectation value of a physically motivated observable $H_{\ms S}$ (or, equivalently, its local Pauli terms). The main virtue of the direct energy witness is thus its simplicity: certifying nonstabilizerness reduces to estimating a single observable. However, this simplicity also limits its scope. Since $W_H$ is a linear witness aligned with the Hamiltonian direction, it can detect only nonstabilizer states whose average energy lies below the stabilizer threshold $E_{\mathrm{STAB}}(H_{\ms S})$. In particular, it is blind to nonstabilizer states whose average energy is the same as that of some stabilizer state, and it becomes trivial whenever the Hamiltonian admits a stabilizer state in its ground space, as happens for commuting Pauli Hamiltonians.

\subsection{Heat measurement as a nonlinear witness}
\label{sec:indirect-heat-measurement}

The direct energy witness is limited by the fact that it depends only on the average energy. As a result, it cannot distinguish nonstabilizer states whose energy expectation value coincides with that of some stabilizer state, and it becomes trivial whenever the stabilizer gap vanishes. To overcome this limitation, we consider the indirect model [see Fig.~\hyperref[F-measurement-models]{\ref{F-measurement-models}(b)}], in which an unknown quantum state interacts with a thermal environment and the resulting heat exchange is measured. As argued in Section~\ref{S:thermodynamic-framework}, quantum features, including nonstabilizerness, leave fingerprints on thermodynamic processes. This can be exploited to construct a witness of nonstabilizerness based solely on energetic data.

To formalize the stabilizer hypothesis conditioned on the measured energy, define the energy slice
\begin{equation}
    \mathrm{STAB}_n(E_0):=\{\sigma \in \mathrm{STAB}_n : E(\sigma) = E_0\}.
\end{equation}
If $\mathrm{STAB}_n(E_0)=\emptyset$, then the energy estimate alone already rules out stabilizer states. Otherwise, we define the stabilizer heat bounds
\begin{align}\label{Eq:stab-heat-window-def}
    Q_c^{\mathrm{STAB}}(E_0) &:= \min_{\sigma \in \mathrm{STAB}_n(E_0)} Q_c(\sigma), \\
    Q_h^{\mathrm{STAB}}(E_0) &:= \max_{\sigma \in \mathrm{STAB}_n(E_0)} Q_h(\sigma).
\end{align}
where $Q_{c/h}(\sigma)$ are the optimal cooling/heating values defined in Eq.~\eqref{Eq:app-optimal-heat}. By construction, any stabilizer state $\sigma\in\mathrm{STAB}_n(E_0)$ and any thermal process must satisfy
\begin{equation}\label{Eq:heat-window-stabilizer-bound}
    Q_c^{\mathrm{STAB}}(E_0) \leq Q \leq Q_h^{\mathrm{STAB}}(E_0).
\end{equation}
Therefore, an experimental observation of $(E_0,Q)$ with $Q<Q_c^{\mathrm{STAB}}(E_0)$ or $Q>Q_h^{\mathrm{STAB}}(E_0)$ certifies that $\rho_{\ms S}\notin\mathrm{STAB}_n(E_0)$ and hence that $\rho_{\ms S}$ is nonstabilizer.

To compute the window boundaries, we note that the optimal heats $Q_{c/h}(\sigma)$ for an individual state $\sigma$ are determined by its nonequilibrium free energy $F_\beta(\sigma)$ by Eq.~\eqref{Eq:QcQh-solution}. For the set $\mathrm{STAB}_n(E_0)$, the extremal values $Q_{c/h}^{\mathrm{STAB}}(E_0)$ are given by the set-witness theorem of Ref.~\cite{de2025heat}, which expresses these bounds in terms of the maximal free energy over the set. Since all states in $\mathrm{STAB}_n(E_0)$ share the same energy $E_0$, maximizing free energy 
\begin{equation}  \label{Eq:stab-max-free-energy}
F_\beta^\star(\mathrm{STAB}_n|E_0):=\max_{\sigma \in \mathrm{STAB}_n(E_0)} F_\beta(\sigma) \!=\! E_0 - \beta^{-1} S_{\min}^{\mathrm{STAB}_n}(E_0),
\end{equation}
is achieved by minimizing the entropy,
\begin{equation}
    S_{\min}^{\mathrm{STAB}_n}(E_0) := \min_{\sigma \in \mathrm{STAB}_n(E_0)} S(\sigma).
\end{equation}
Let $\beta_c^\star\le \beta_h^\star$ be the smallest and largest solutions (when they exist) of the root equation $F_\beta[\gamma_{\ms S}(x)] = F_\beta^\star(\mathrm{STAB}_n|E_0)$ with boundary conventions as in Section~\ref{S:thermodynamic-framework} when a branch has no finite solution. Then the stabilizer heat window boundaries follow from Eq.~\eqref{Eq:QcQh-solution}:
\begin{align}
    Q_{c/h}^{\mathrm{STAB}}(E_0) = E_0 - E[\gamma_{\ms S}(\beta_{c/h}^\star)],
\end{align}
where $E(\rho):=\tr(H_{\ms S}\rho)$. This establishes heat exchange as a nonlinear witness of nonstabilizerness, based solely on energetic data $(E_0,Q)$.

In practice, both $E_0$ and $Q$ are estimated with finite statistical error. Denoting the estimates by $\hat E_0$ and $\hat Q$ with the respective uncertainties $\delta E$ and $\delta Q$, the witnessing condition becomes robust if $\hat Q \pm \delta Q$ lies entirely outside the interval $[Q_c^{\mathrm{STAB}}(\hat E_0 \pm \delta E), Q_h^{\mathrm{STAB}}(\hat E_0 \pm \delta E)]$. Alternatively, a confidence interval for the witness violation may be computed using standard statistical methods. Note that our witnessing protocol is sound: any violation (with appropriate statistical confidence) certifies nonstabilizerness. However, in general, it will provide only a necessary condition for stabilizerness, as the heat bounds might be respected even if the state lies outside the stabilizer polytope.

\section{Stabilizer energy witness \label{S:stabilizer-energy-witness}}
In this section, we analyze the direct energy measurement scenario illustrated in Fig.~\hyperlink{F-measurement-models}{\ref{F-measurement-models}(a)}. For notational simplicity, we suppress the system label and denote the system by $(\mathcal H,H)$, where $\mathcal H=(\mathbb C^{2})^{\otimes n}$ and $H\in\mathrm{Herm}(\mathcal H)$ as Sec. \ref{S:stabilizer-resource-theory}. The central question is: \emph{What are the Hamiltonians $H$ whose energy estimate $\tr(H\rho)$ can witness the nonstabilizerness of an unknown state $\rho$?}

\subsection{Stabilizer gap and witnessing power}

As described in Sec.~\hyperref[SS:direct-measurement]{IV-1}, Hamiltonians with $\delta_{\mathrm{STAB}}(H)>0$ are linear witnesses for magic. A distinguished family of such witnesses arises from the set of facets $\mathrm F_n$ defining the $H$-representation of $\mathrm{STAB}_n$ in Eq. (\ref{eq:Hrepstab}). If $A \in \mathrm F_n$ is non-trivial, meaning that it is not positive semidefinite, we guarantee that for every quantum state $\rho$:
\begin{equation}
    \tr(A\rho) < 0 \Rightarrow \delta_\mathrm{STAB}(A) >0\;,
\end{equation}
since all stabilizer states attain nonnegative values in expectation. For small systems where the facet set $\mathrm{F}_n$ is explicitly known ($n=1,2$ \cite{zurel2020hidden,Palhares2026}), these operators provide optimal linear witnesses. 

\begin{mybox}{Optimal single-qubit witness}
For a normalized Hamiltonian $H=\mathbf h\cdot\v\sigma$ with $\|\mathbf h\|_2=1$, the stabilizer gap is
\begin{equation}\label{Eq:1q-gap}
    \delta_{\mathrm{STAB}}(H)=1-\max_{i\in\{x,y,z\}}|h_i|.
\end{equation}
The maximum gap is $1-\tfrac{1}{\sqrt3}$, attained for $H=-\tfrac{1}{\sqrt3}(\v f\cdot\v\sigma)$ with $\v f\in\{\pm1\}^3$. These Hamiltonians are affinely related to the facet operators $1-\v f\cdot\v\sigma$, since $1-\v f \cdot \v\sigma = \iden +\sqrt{3}H$. Thus they define the same witnessing direction, up to an additive identity shift and a positive rescaling. These optimal Hamiltonians correspond to the Clifford orbit of the $T$-type magic state 
\begin{equation}
    \ketbra{T}=\tfrac12\qty[\iden+\tfrac{1}{\sqrt3}(X+Y+Z)].
\end{equation}
\end{mybox}

For non‑interacting Hamiltonians $H=\sum_{i=1}^n \v{h}_i\cdot\v{\sigma}_i$, the stabilizer optimization factorizes and the gap is additive:
\begin{equation}
    \delta_{\mathrm{STAB}}(H)=\sum_{i=1}^n\Bigl(\|\v{h}_i\|_2-\max\{|h_{i,x}|,|h_{i,y}|,|h_{i,z}|\}\Bigr).
\end{equation}
Under normalization $\|\v{h}_i\|_2=1$ for each $i$, the maximum gap per ‐ site is $1-\tfrac{1}{\sqrt{3}}$ and the maximum total gap is $n\qty(1-\tfrac{1}{\sqrt{3}})$. For interacting ($k\ge2$) Hamiltonians, vanishing gaps can occur even if they are non‑stabilizer. It is illustrative to consider the following two-qubit example:

\begin{mybox}{Fine-tuned two-qubit witness}
Consider the one‑parameter family
\begin{equation}
    H_{\Phi^+}(\varepsilon)=-X\otimes X - Z\otimes Z + \varepsilon\,(Z\otimes\iden-\iden\otimes Z),
    \label{Eq:bell-Hamiltonian}
\end{equation}
with $\varepsilon\ge0$. At $\varepsilon=0$, $H_{\Phi^+}(0)$ is a stabilizer Hamiltonian with unique ground state $\ket{\Phi^+}=2^{-\nicefrac{1}{2}}(\ket{00}+\ket{11})$. For $\varepsilon>0$ the perturbation anticommutes with $X\otimes X$, so the Hamiltonian is no longer of the stabilizer form. A direct diagonalization gives
\begin{align}
    E_{\mathrm{gs}}\bigl[H_{\Phi^+}(\varepsilon)\bigr]&=\min\bigl\{-2,\;1-\sqrt{1+4\varepsilon^2}\bigr\},\\
    E_{\mathrm{STAB}}\bigl[H_{\Phi^+}(\varepsilon)\bigr]&=\min\bigl\{-2,\;1-2\varepsilon\bigr\}.
\end{align}
Because $\bra{\Phi^+}H_{\Phi^+}(\varepsilon)\ket{\Phi^+}=-2$, the Bell state remains a ground state for $0\le\varepsilon\le\sqrt2$ (with a level crossing at $\varepsilon=\sqrt2$). Hence $\delta_{\mathrm{STAB}}[H_{\Phi^+}(\varepsilon)]=0$ in this interval, even though $H_{\Phi^+}(\varepsilon)$ is not a stabilizer Hamiltonian. This demonstrates that vanishing stabilizer gaps can occur for non‑stabilizer Hamiltonians through fine‑tuning. For $\varepsilon > \sqrt 2$, there is already a stabilizer gap: To illustrate what kind of states have their magic witnessed in this regime, consider the following family of states:
\begin{equation}
    \ket{\psi(\theta, \phi)} = \iden \otimes e^{-i\phi \tfrac{X}{2}}[\cos\theta \ket{00} + \sin \theta \ket{11}]\;.
    \label{Eq:2qbitexample}
\end{equation}
Since every two qubit pure state can be rotated to $\cos\theta \ket{00} + \sin \theta\ket{11}$ with $\theta \in [0, \tfrac\pi2]$ by local unitaries, it is natural to consider the two-parameter family of states labeled by $(\theta, \phi)$, defined on a half-sphere. In Fig.~\ref{fig:energy_2qbitexample}, we plot the ground and stabilizer energies, and highlight the yellow region where nonstabilizerness is witnessed. In particular, for $\varepsilon=\tfrac32$, we also highlight the states that have energies in the corresponding window.
\end{mybox}

\begin{figure}
    \centering
    \includegraphics{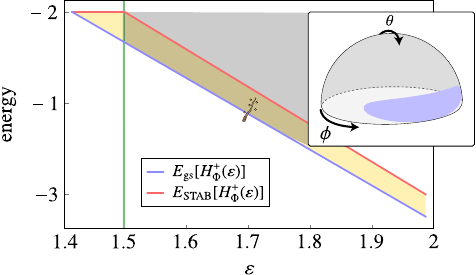}
    \caption{\emph{Stabilizer-gap for the perturbed Bell Hamiltonian}. Comparison of the true ground-state energy $E_{\mathrm{gs}}[H_{\Phi^+}(\varepsilon)]$ and the stabilizer ground-state energy $E_{\mathrm{STAB}}[H_{\Phi^+}(\varepsilon)]$ for the two-qubit Hamiltonian $H_{\Phi^+}(\varepsilon)$ in Eq.~\eqref{Eq:bell-Hamiltonian}, shown for $\varepsilon \in (\sqrt{2},2]$. The yellow region between the two curves marks the energy window in which nonstabilizerness is certified. At $\varepsilon = \tfrac32$ (green vertical line), the blue region on the $(\theta,\phi)$ half-sphere indicates the family of states in Eq.~\eqref{Eq:2qbitexample} whose energies fall below $E_{\mathrm{STAB}}$, and are therefore witnessed to be nonstabilizer.}
    \label{fig:energy_2qbitexample}
\end{figure}

However, in general, perturbing a stabilizer Hamiltonian creates a positive gap. The following lemma gives a sufficient condition.

\begin{lem}[Perturbation creates a stabilizer gap] \label{lem:perturb-stabilizer}
Let $H_0$ be a stabilizer Hamiltonian with a unique ground state $\ket{\mathbf S}$ and a spectral gap $\Delta>0$. Consider the perturbed Hamiltonian $H(\lambda)=H_0+\lambda V$ with $V^\dagger=V$. Assume that $V$ couples the ground state with excited states, i.e., $(\iden-\ketbra{\mathbf S})V\ket{\mathbf S}\neq0$. Then there exists $\lambda_0>0$ such that for all $0<|\lambda|<\lambda_0$,
\begin{enumerate}
    \item $H(\lambda)$ has a unique ground state $\ket{\psi(\lambda)}$ that is not a stabilizer state;
    \item $\delta_{\mathrm{STAB}}(H(\lambda))>0$.
\end{enumerate}
\end{lem}
\begin{proof}
For sufficiently small $|\lambda|$ the ground state remains unique with the energies and the eigenprojector $P(\lambda) = \ket{\psi(\lambda)}\bra{\psi(\lambda)}$ depending analytically on $\lambda$~\cite{kato1966perturbation}. The proof follows by contradiction: If $P(\lambda)$ were a stabilizer state for a sequence $\lambda_m\to0$, then by finiteness of the set of pure stabilizer states there would exist a subsequence with $P(\lambda_{m_k}) = \ketbra{\phi} \in \mathrm{STAB}_n\;\forall k$. Continuity at $\lambda=0$ imposes:
\begin{equation}
\ketbra{\phi} = \lim_{k \to \infty} P(\lambda_{m_k}) = \ketbra{\mathbf S}\;.
\end{equation}
This implies that $\ket{S}$ must satisfy the eigenvalue equation for every element of the subsequence, $(H_0 + \lambda_{m_k}V)\ket{S} =E(\lambda_{m_k})\ket{S}$. Denoting its eigenvalue with $H_0$ as $E_0$, since we have $\lambda_{m_k} \neq 0$ by construction:
\begin{equation}
    V \ket{\mathbf S} = \left(\frac{E(\lambda_{m_k})-E_0}{\lambda_{m_k}}\right)\ket{\mathbf S}\propto \ket{\mathbf S}\;,
\end{equation}
contradicting the coupling assumption. Hence $P(\lambda)$ is nonstabilizer, and because the ground state is unique, no stabilizer state can achieve $E_{\mathrm{gs}}(H(\lambda))$, so $E_{\mathrm{STAB}}(H(\lambda))>E_{\mathrm{gs}}(H(\lambda))$.
\end{proof}
In fact, a positive stabilizer gap can be shown even if the stabilizer Hamiltonian $H_0$ is degenerate, as shown in the Appendix~\ref{app:degeneratestabilizergap}. In that context, we also need to assume that there is a state $\ket{\mathbf  S}$ in the ground state subspace such that $\langle V^2\rangle_{\mathbf S}\neq 0$, along with another technical condition. Then, one can show that for sufficiently small $\lambda$, there is a positive stabilizer gap and the state
\begin{equation}
\ket{\psi(\lambda)} =\frac{\ket{\mathbf S}-\lambda \Pi_{\mathbf S^\perp} V \ket{\mathbf S}}{\sqrt{1+\langle V^2\rangle}_\mathbf S}\;,
\end{equation}
where $\Pi_{\mathbf S^\perp}$ is the projector in the orthogonal complement of the ground space, has energy smaller than $E_\mathrm{STAB}(H)$, giving an energy witnessed by such gap.

We note that the family $H_{\Phi^+}(\varepsilon)$ bypasses this lemma because the perturbation $V=\varepsilon(Z\otimes\iden-\iden\otimes Z)$ satisfies $V\ket{\Phi^+}=0$, failing the coupling condition. However, if we instead consider the perturbation $V^\prime= \varepsilon_1 Z \otimes \iden + \varepsilon_2 \iden \otimes Z$ to $H_{\Phi^+}(0)$, we note that:
\begin{equation}
V^\prime \ket{\Phi^{+}} =(\varepsilon_1 + \varepsilon_2)\ket{\Phi^-}\;,
\end{equation}
with $\ket{\Phi^-} =2^{-\nicefrac{1}{2}}(\ket{00}-\ket{11})$. Since $\langle \Phi^-  \ket{\Phi^+}=0$, this new perturbation now satisfies the coupling condition for $\varepsilon_1\neq \varepsilon_2$, and indeed the Hamiltonian has a nonstabilizer  ground state. This perspective makes the fine-tuning in the two-qubit example of Eq. (\ref{Eq:bell-Hamiltonian}) transparent: If the perturbation is generic enough, it usually satisfies our sufficient conditions, guaranteeing a stabilizer gap.

For $k\geq 2$, vanishing stabilizer gaps can occur well beyond perturbations of the stabilizer-Hamiltonian class. A representative example is the $n$-qubit ferromagnetic Heisenberg chain (with periodic boundary conditions $n+1\equiv 1$),
\begin{equation}\label{Eq:heisenberg-chain}
    H_{\text{Heis}}:= -\sum_{j=1}^n (X_j X_{j+1}+Y_{j}Y_{j+1}+Z_j Z_{j+1}).
\end{equation}
The terms in Eq.~\eqref{Eq:heisenberg-chain} do not commute, so $H_{\text{Heis}}$ is not of stabilizer form~\eqref{Eq:stabilizer-Hamiltonian}. Nevertheless, $\ket{0}^{\otimes n}$ is a ground state. Indeed, for two qubits the operator $X\otimes X+Y\otimes Y+ Z\otimes Z$ has spectrum $\{+1, -3\}$, hence obeying the operator inequality
\begin{equation}
    X_j X_{j+1}+Y_{j}Y_{j+1}+Z_j Z_{j+1} \leq \iden,
\end{equation}
so each bond term satisfies $-(\v\sigma_j\cdot\v\sigma_{j+1})\ge -\iden$. Summing over $j$ yields the global bound $H_{\mathrm{Heis}}\ge -n\,\iden$, i.e. $E_{\mathrm{gs}}(H_{\mathrm{Heis}})\ge -n$. On the other hand, $\ket{0}^{\otimes n}$ saturates this bound because each neighboring pair $\ket{00}$ lies in the symmetric (triplet) subspace where $\v\sigma_j\cdot\v\sigma_{j+1}=+1$, giving $\bra{0^{\otimes n}}H_{\mathrm{Heis}}\ket{0^{\otimes n}}=-n$. Therefore $E_{\mathrm{gs}}(H_{\mathrm{Heis}})=-n$ and, since $\ket{0}^{\otimes n}\bra{0}^{\otimes n}\in\mathrm{STAB}_n$,
\begin{equation}
    E_{\mathrm{STAB}}(H_{\mathrm{Heis}})=E_{\mathrm{gs}}(H_{\mathrm{Heis}})
    \quad\Longrightarrow\quad
    \delta_{\mathrm{STAB}}(H_{\mathrm{Heis}})=0.
\end{equation}
This illustrates that $\delta_{\mathrm{STAB}}(H)=0$ can occur for interacting, nonstabilizer  Hamiltonians whenever the ground space contains a stabilizer state.

\vspace{1cm}
\begin{mybox}{Interacting Hamiltonian with vanishing stabilizer gap}
\, The ferromagnetic Heisenberg chain~\eqref{Eq:heisenberg-chain} is not a stabilizer Hamiltonian, yet it has a stabilizer ground state $\ket{0}^{\otimes n}$. Consequently, $\delta_{\mathrm{STAB}}(H_{\mathrm{Heis}})=0$, so $H_{\mathrm{Heis}}$ cannot witness nonstabilizerness via ground-state energy measurements. This shows that vanishing stabilizer gaps are not exclusive to stabilizer Hamiltonians.
\end{mybox}

This Hamiltonian bypass previous results since it cannot be immediately written as a stabilizer Hamiltonian with a small perturbation added. It also illustrates that determining the stabilizer gap for interacting many-qubit Hamiltonians might be non-trivial, and necessary to be explicitly computed to determine if a given Hamiltonian can witness nonstabilizerness.

\subsection{Algorithms for the stabilizer energy}

The relevant question is now algorithmic: \emph{Given a $k-$local Hamiltonian $H$ as promised by the setup in Sec. \ref{S:measurement-witnessing}, how can $E_\mathrm{STAB}(H)$ and the corresponding stabilizer gap be obtained?} and can they be obtained \emph{efficiently}, meaning that there is an algorithm with runtime polynomial in the number of qubits as $n\to \infty$? In complexity theory, computing $E_{\mathrm{gs}}(H)$ for a general $k$‑local Hamiltonian is QMA‑hard~\cite{kitaev2002classical}, being a hard task even for quantum computers, and thus the hope of efficiently obtaining the exact stabilizer gap in Eq. (\ref{Eq:stabilizer-gap}) is limited, with the most generic strategy being (sparse-) diagonalizing the $2^n\times 2^n$ Hamiltonian matrix $H$. In order to use $\tr(H\rho)$ as a witness, it suffices to know: (1) the stabilizer ground state $E_\mathrm{STAB}(H)$, and (2) some quantum state satisfying Eq. (\ref{eq:witness-stab}), we have the guarantee of a positive stabilizer gap, which can be done without diagonalizing the Hamiltonian, as, for example, applying Lemma \ref{lem:perturb-stabilizer} (or its degenerate version in App. \ref{app:degeneratestabilizergap}, Lemma \ref{lem:degperturbstabilizer}) if the conditions apply, or by numerically find a variational state with energy lower than the stabilizer value.

However, a naive evaluation of $E_\mathrm{STAB}(H)$ is hard: Notice that we can write:
\begin{equation}
E_\mathrm{STAB}(H) =\min_{\rho \in \mathrm{STAB}_n}\tr(H\rho)  =\min_{\mathbf S \in \mathrm{ext}(\mathrm{STAB}_n)}\langle \mathbf S| H |\mathbf S\rangle\;,
\label{eq:Estab-hard}
\end{equation}
where we have used the fact that optimization over the stabilizer polytope can be restricted to its vertices due to convexity. Since it is known that $|\mathrm{ext}(\mathrm{STAB}_n)|=2^{\Theta(n^2)}$ \cite{aaronson2004improved}, the naive strategy would be to compute the energy of each vertex of the stabilizer polytope and take the minimal one, at least taking superexponential time. In \cite{sun2025stabilizerground}, the structure of $k$-local Hamiltonians was exploited to restrict this optimization. It was shown that Eq. (\ref{eq:Estab-hard}) actually restricts and can be written as:
\begin{equation}
E_\mathrm{STAB}(H) =-\max_{\mathbf Q \in C_\mathrm{max}[\mathbf P(H)]}\sum_{P \in \mathbf Q}w_P\;,
\label{eq:Estab-restricted}
\end{equation}
where $C_\mathrm{max}[\mathbf P(H)]$ is the set of maximal commuting subsets of $\pm \mathbf P(H)$ that can form stabilizer groups.
In App.~\ref{app:stabgroundenergy} we review the formal statement and the corresponding proof. By applying this restriction in optimization to commuting subsets, one can find the stabilizer ground state energy in $O(2^{c n \log n})$ time, with $c=O(1)$. In general, it will push the complexity of finding the minimal stabilizer energy into this combinatorial problem of finding independent sets of Paulis that have high weights, which depends on how complex the algebraic structure of $C_\mathrm{max}[\mathbf P(H)]$ is. 

There is an elegant way to illustrate this concept. Let $G(H)$ be the anti-commutation/frustration graph of $\mathbf P(H)$~\cite{gokhale2019minimizing,chapman2020characterization,mann2025graph,xu2025simultaneous}, which means that it is the graph whose vertices are $\mathbf P(H)$ and the edge set defined as:
\begin{equation}
    E[G(H)] =\{(P_1, P_2) \in \mathbf P(H)^{\times 2}|\;P_1 P_2=-P_2P_1\}\;.
\end{equation}
Then, the elements of $C_\mathrm{max}[\mathbf P(H)]$
are interpreted as maximal independent sets on the anticommutation graph. See Fig.~\ref{fig:frustrationgraph-example} for the examples of the Heisenberg and the transverse-field Ising chain, defined as:
\begin{equation}
    H_\mathrm{Ising} =-\sum_{j=1}^n (Z_j Z_{j+1}+ h X_j)\;,
    \label{eq:Ising-Hamiltonian}
\end{equation}
also taken with periodic boundary conditions. In general, $C_\mathrm{max}[\mathbf P(H)] \subseteq I_\mathrm{max}[G(H)]$, where $I_\mathrm{max}[G(H)]$ is the set of (maximal) independent subsets of $G(H)$, that is, subset of vertices that do not contain any edges and cannot increase. However, not every independent subset lifts to an element of $C_\mathrm{max}[\mathbf P(H)]$, due to algebraic dependence of the Paulis: For example, in the Heisenberg chain, there is a $S \in I_\mathrm{max}[G(H)]$ such that $\{Z_1Z_2, Y_1Y_2, X_1X_2\} \subseteq S$, since they form a commuting tuple. However, they cannot generate a stabilizer group, as their product is $-1$.
\begin{figure}[t]
    \centering
    \includegraphics{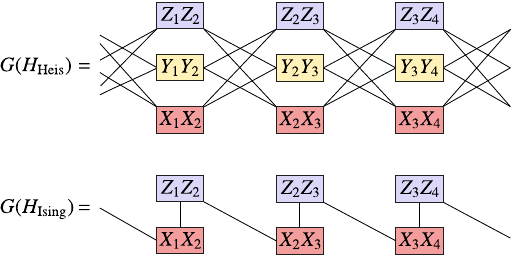}
    \caption{\emph{Anti-commutation (frustration) graphs associated with two local Hamiltonians}. Vertices denote Pauli terms and edges represent anti-commuting pairs. The Heisenberg Hamiltonian (top) [Eq.~\eqref{Eq:heisenberg-chain}] gives a highly connected graph due to the non-commuting $X$, $Y$, and $Z$ interactions on each edge, while the transverse-field Ising model (bottom) [Eq.~\eqref{eq:Ising-Hamiltonian}] produces a bipartite graph separating $Z Z$ and $X$ terms. This distinction directly impacts the structure of maximal commuting subsets and the complexity of evaluating $E_\mathrm{STAB}(H)$.}
    \label{fig:frustrationgraph-example}
\end{figure}
This extra algebraic dependency on Pauli operators is one of the sources of hardness in the computation of the commuting subsets. However, we can introduce a class of Hamiltonians where the problem of commuting sets reduces to finding independent sets on the frustration graph:  
Let $\{\mathbf S_j\}_{j=1}^\ell$ be a list of stabilizer groups, with corresponding stabilizer Hamiltonians $\{H_{\mathbf S_j}\}_j$ given as Eq.~\eqref{Eq:stabilizer-Hamiltonian}. We define
\begin{equation}
    H \equiv \sum_{j=1}^{\ell}w_j H_{\mathbf S_j}\;.
    \label{Eq:stabilizerdecompHamiltonian}
\end{equation}
We refer to this class of Hamiltonians, now parametrized by $\{w_j\}_{j=1}^\ell$, by sum-of-stabilizers. For this class, estimation of the stabilizer ground state energy does turn out to be a purely graph-theoretical problem:
\begin{lem}[Stabilizer Energy via MWIS]
    Let $\{\mathbf S_j\}_{j=1}^\ell$ be a set of independent stabilizer groups, which means $\mathbf S_j \cap \langle \cup_{k \in S} \mathbf S_k \rangle = \{\iden\}$   for all $j \in [\ell], S \subseteq [\ell]$, and $H$ be a corresponding sum-of-stabilizers Hamiltonian. Then:
    \begin{equation}
        E_\mathrm{STAB}(H) = -\max_{\mathbf Q \in I_{\max}[G(H)]} \sum_{P \in \mathbf Q}w_P\;.
        \label{Eq:energysumofstabilizers}
    \end{equation}
    This is referred to as computing the maximal weight independent set of the graph $G(H)$ with weights $\{w_P\}_{P \in \mathbf P(H)}$. \label{lem:sumofstabilizers}
\end{lem}
The proof is presented in Appendix~\ref{app:proofsumofstabilizers}. It is known that the MWIS problem is NP-hard, although there are algorithms that can compute MWIS on very large graphs with millions of vertices~\cite{lamm2019exactly}.
\begin{figure}[t]
    \centering
    \includegraphics{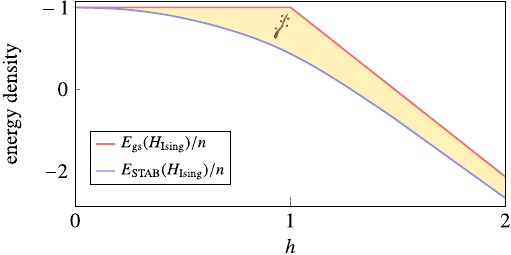}
    \caption{\emph{Quantum and stabilizer ground-state energy densities for the transverse-field Ising chain}. The exact energy density $E_\mathrm{gs}(H)/n$ (red) and its stabilizer counterpart $E_\mathrm{STAB}(H)/n$ (blue) are plotted as functions of the transverse field $h$ for $n=100$. The shaded region corresponds to energies that certify nonstabilizerness. The stabilizer gap vanishes in the limits $h \to 0$ and $h \to \infty$, and is maximal near the quantum critical point at $h=1$.}
    \label{fig:tfimstabilizergap}
\end{figure}
Fortunately, for various classes of hamiltonains, the estimation of $E_\mathrm{STAB}(H)$ is much easier than the worst-case upper bounds mentioned above. In \cite{sun2025stabilizerground}, it was shown that if one restricts to one-dimensional Hamiltonians, there is an algorithm with $O(n)$ runtime to compute its stabilizer ground energy, for example. Furthermore, if its anti-commutation graph is bipartite, it is known that MWIS can be solved in polynomial time \cite{hopcroft1973n}. Let $\mathbf P_X$ and $\mathbf P_Z$ be two Pauli sets fully composed of $X$ and $Z$ products, respectively. The corresponding family of Hamiltonians:
\begin{equation}
    H(\mathbf P_X, \mathbf P_Z) =-\sum_{P \in \mathbf P_Z}P -w \sum_{P^\prime \in \mathbf P_X}P^\prime\;,
    \label{Eq:bipartiteHamiltonian}
\end{equation}
is sum-of-stabilizers, with stabilizer groups $\mathbf S_X =\langle \mathbf P_X\rangle$, $\mathbf S_Z =\langle \mathbf P_Z\rangle$ independent, satisfying the assumptions of Lemma~\ref{lem:sumofstabilizers}. Hence, for this class of \emph{bipartite Hamiltonians}, the witness can be efficiently evaluated.

For example, consider the transverse-field Ising chain, as defined in Eq.~\eqref{eq:Ising-Hamiltonian}, with $h>0$. In this case, we have $\mathbf P_Z = \langle \{Z_j Z_{j+1}\}_{j=1}^n \rangle$ and $\mathbf P_X = \langle \{X_j\}_{j=1}^n \rangle$. Given $\mathbf Q \in C_\mathrm{max}[\mathbf P(H)]$, with stabilizer energies $-n$ and $-nh$, respectively. Notice that $h > 1$ or $h < 1$ selects one of two stabilizer groups, and in the two regimes, all the other maximal commuting subset has a higher energy. Therefore,
\begin{equation}
    E_\mathrm{STAB}(H_\mathrm{ising})  = -n \max(1,h)\;.
\end{equation}
Due to integrability, the exact analytical expression of $E_\mathrm{gs}(H_\mathrm{ising})$ is known for any $n$. In Fig.~\ref{fig:tfimstabilizergap}, we plot both the ground state energy and its stabilizer restriction for $n=100$. It has an interesting behavior: The stabilizer gap vanishes in two regimes: when $h=0$, when the ground space is the code of $\mathbf S_Z$, and when $h \to \infty$, reducing to the code of $\mathbf S_X$, and its maximum appears at $h=1$, where it is known to host a quantum phase transition~\cite{sachdev2011quantum}. This is somewhat expected: It is believed that at the quantum critical point, the state has \emph{long-range magic}~\cite{wei2025long,white2021conformal}, meaning that it can only be related to a stabilizer state by a deep quantum circuit; it is not expected that its energy can also be approximated by a stabilizer state. 

This example shows that the energy of a quantum state not only has enough information to witness nonstabilizerness, but the stabilizer gap behavior itself within a Hamiltonian family class can reveal interesting features about the ground state structure.

\section{Stabilizer heat witness \label{S:stabheatwit}}

We begin with the single-qubit case, where the heat-based witness admits a complete geometric characterization. For fixed energy, the witness detects exactly those states whose entropy is lower than that of every stabilizer state compatible with the same energy. This gives an if-and-only-if criterion for detectability and also a condition for when the witness is optimal along a given noisy family. The necessary and sufficient condition for detecting nonstabilizerness is captured by the following Theorem:

\begin{thm}[Single-qubit heat-detectability criterion]\label{Thm:heat-detectability}
Let $\rho=\frac12(\iden+\v r\cdot \v \sigma)$ be a qubit state with Hamiltonian $H=\v h\cdot \v \sigma$, where $\|\v h\|_2=1$, and let $E_{\v r}:=\tr(\rho H)=\v h\cdot \v r$. Assume $\mathrm{STAB}_1(E_{\v r})\neq\emptyset$. Then the heat-based nonstabilizer witness detects $\rho$ if and only if
\begin{equation}\label{Eq:detection-equation}
    \|\v r\|_2 >R_\star(E_{\v r})
\end{equation}
where $R_\star(E_{\v r}):=\max\Bigl\{\|\v s\|_2:\ \|\v s\|_1\le 1,\ \v h\cdot \v s=E_{\v r}\Bigr\}.$
\end{thm}
\begin{sproof}
By construction of the heat-based witness, detectability on the fixed-energy slice $E_{\v r}$ is equivalent to the nonequilibrium free energy of $\rho$ exceeding the maximal free energy attainable by stabilizer states with the same energy. Since every $\sigma\in \mathrm{STAB}_1(E_{\v r})$ has energy $E_{\v r}$, this condition reduces to $S_{\min}^{{\textrm{STAB}_1}}(E_{\v r})$. For qubits, the entropy depends only on the Bloch radius and is strictly decreasing in $\|\v r\|_2$, namely $  S(\rho)=H_2\qty(\frac{1+\|\v r\|_2}{2})$. Therefore, minimizing the entropy over $\mathrm{STAB}_1(E_{\v r})$ is equivalent to maximizing the Bloch radius over the stabilizer slice $\{\v s:\|\v s\|_1\le 1,\ \v h\cdot \v s=E_{\v r}\}$. Hence $S_{\min}^{{\textrm{STAB}_1}}(E_{\v r})$ if and only if $\|\v r\|_2 > R_\star(E_{\v r})$, which proves the claim. Full details are given in Appendix~\hyperref[App:single-qubit-witness]{A-2}.
\end{sproof}

A natural question is whether the heat-based witness certifies nonstabilizerness for all states in a given noisy family up to the true stabilizer threshold. The following corollary answers this question.
\begin{cor}[Optimality]\label{Cor:optimality}
Let $\{\rho_\lambda\}_{\lambda\in[0,1]}$ be a single-qubit family with fixed energy $\tr(\rho_\lambda H)=E_0$, and suppose that $S(\rho_\lambda)$ is strictly increasing in $\lambda$. Define
\begin{equation}
    \lambda_{\star}:=\inf\{\lambda:\rho_\lambda\in \mathrm{STAB}_1\}.
\end{equation}
Then the heat-based witness is tight on this family if and only if $S(\rho_{\lambda_{\star}})= S_{\min}^{\mathrm{STAB}_1|E_0}$. Equivalently,
\begin{equation}\label{Eq:optimality-eq}
    \|\v r_{\lambda_{\star}}\|_2 = \max\Bigl\{\|\v s\|_2:\ \|\v s\|_1\le 1,\ \v h\cdot \v s=E_0\Bigr\}.
\end{equation}
\end{cor}
\begin{sproof}
Since $\tr(\rho_\lambda H)=E_0$ for all $\lambda$, Theorem~\ref{Thm:heat-detectability} applies on the same fixed-energy slice throughout the family and gives $\rho_\lambda$ is detected if and only if $S(\rho_\lambda)<S_{\min}^{\mathrm{STAB}_1|E_0}$. Because $S(\rho_\lambda)$ is strictly increasing in $\lambda$, the detected states form an initial interval in $\lambda$, ending at the unique point where the entropy reaches the stabilizer minimum on that slice. Hence the witness is tight if and only if this happens exactly at the true stabilizer threshold $\lambda_\star$, namely $S(\rho_{\lambda_\star})=S_{\min}^{\mathrm{STAB}_1|E_0}$. Using again that, for qubits, the entropy is a strictly decreasing function of the Bloch radius, this is equivalent to $\|\v r_{\lambda_\star}\|_2 = R_\star(E_{\v r})$, which proves the corollary. Full details are given in Appendix~\hyperref[App:single-qubit-witness]{A-2}.
\end{sproof}

To gain some intuition about the previous results, consider a family of states with fixed average energy defined as a mixture of a magical state and a stabilizer state:
\begin{equation}\label{Eq:family-state}
    \rho_\lambda = (1-\lambda) \rho_{\text{magic}} + \lambda\rho_{\text{stab}} \:\:\: \text{with}\:\:\: \tr(\rho_\lambda H) = E_0.
\end{equation}
where $H$ is a non-interacting Hamiltonian and $\rho_{\text{magic}}$ is a magical state and $\rho_{\text{stab}} \in \textrm{STAB}_n$.  
\begin{figure}[t]
    \centering
    \includegraphics{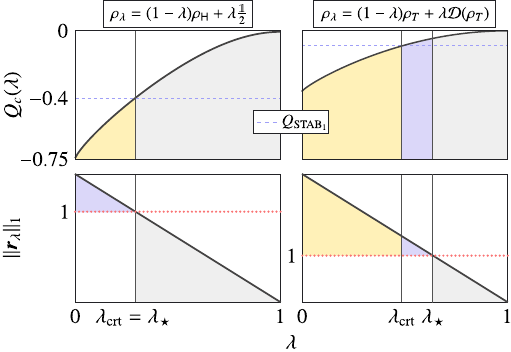}
    \caption{\emph{Magic detection via heat exchange for single-qubit noisy families.} Top row: optimal cooling heat $Q_c(\rho_\lambda)$ (solid black curve) as a function of the noise parameter $\lambda$, for a qubit coupled to a thermal environment at inverse temperature $\beta=1$. The dashed blue horizontal line marks the stabilizer heat threshold $Q_{\mathrm{STAB}1}$. The yellow region indicates the values of $\lambda$ for which the heat-based witness certifies nonstabilizerness, namely when $Q_c(\rho_\lambda)<Q_{\mathrm{STAB}1}$. The vertical line at $\lambda_{\mathrm{crt}}$ marks the point where the witness ceases to detect magic. Bottom row: the stabilizer criterion with the red dotted line marking the boundary $\|\v r_\lambda\|_1=1$. The vertical line at $\lambda_\star$ denotes the true stabilizer threshold, so that states are nonstabilizer for $\lambda<\lambda_\star$. In the left panel, $\lambda_{\mathrm{crt}}=\lambda_\star$, showing that the heat-based witness is optimal along the family. In the right panel, $\lambda_{\mathrm{crt}}<\lambda_\star$, so the witness is not optimal: there is an intermediate interval in which the state remains nonstabilizer but is no longer detected by heat exchange.}
    \label{F:heat-singlequbit-plot}
\end{figure}
We now compare two one-parameter single-qubit families for which direct energy measurements alone are inconclusive. In the first, the Hamiltonian has a magic ground state and a positive stabilizer gap, yet the family remains on a fixed-energy slice and therefore cannot be distinguished by average energy alone. In the second case, the Hamiltonian has stabilizer eigenstates and a vanishing stabilizer gap. Whenever direct energy measurement is not conclusive, one can instead turn to the heat-based witness, which requires only knowledge of the system’s average energy together with the minimum entropy attainable within the stabilizer polytope at that energy.

These two features, which are not captured by direct energy measurement but are revealed by the heat-based witness, are illustrated in Fig.~\ref{F:heat-singlequbit-plot}. Here, we consider the optimal heat exchange between a single qubit and a thermal environment, focusing on the case in which the system warms up while the environment cools down. Analogous bounds can be derived for the opposite branch, in which the system cools down while the environment heats up. First, we consider a qubit described by the Hamiltonian $H_1 = \v h_{1} \cdot \v \sigma$, where $\v h_1 \propto (1,-1,-1)$, whose eigenstates are nonstabilizer states [top row and left panel of Fig.~\ref{F:heat-singlequbit-plot}]. In particular, $H_1$ has a magic ground state and a positive stabilizer gap. We take the depolarized $H$-state family $\ketbra{\ms H} = \tfrac12(\iden + \tfrac{X+Z}{\sqrt{2}})$. Since $\tr(\rho_\lambda H_1)=0$ for all $\lambda$, the direct energy witness cannot distinguish states within this family, despite the fact that $H_1$ itself has a magic ground state. The Bloch vector of $\rho_\lambda$ is $\v r_\lambda=(1-\lambda)\Bigl(\tfrac{1}{\sqrt2},0,\tfrac{1}{\sqrt2}\Bigr)$, so that $\|\v r_\lambda\|_1=\sqrt2(1-\lambda)$ and $\|\v r_\lambda\|_2=1-\lambda$, hence the true stabilizer threshold is $\lambda_\star=1-\tfrac{1}{\sqrt2}$ [bottom row and left panel of Fig.~\ref{F:heat-singlequbit-plot}]. On the other hand, the family lies on the slice $E_0=0$, for which $R_\star(0)=2^{-1/2}$. Therefore, the heat-detection threshold is $\lambda_{\mathrm{crt}}=1-2^{-1/2}=\lambda_\star$. Thus, the heat-based witness is optimal along this family: it certifies nonstabilizerness up to the exact point where the state enters the stabilizer polytope.

Next, we consider the Hamiltonian $H_2=\sigma_z$ whose eigenstates are stabilizer states, so that $\delta_{\mathrm{STAB}}(H_2)=0$. We then take the dephased noisy $T$-state family $\rho_\lambda=(1-\lambda)\rho_T+\lambda\,\mathcal D(\rho_T)$, with $\rho_T=\tfrac12\Bigl(\openone+\tfrac{X+Y+Z}{\sqrt3}\Bigr)$, and $\mathcal D$ denotes dephasing in the eigenbasis of $Z$. In this case, $\tr(\rho_\lambda H_2)=\tfrac{1}{\sqrt3}$ for all $\lambda$, so once again the direct energy witness is inconclusive along the family. The corresponding Bloch vector is $\v r_\lambda= \Bigl(\tfrac{1-\lambda}{\sqrt3},\tfrac{1-\lambda}{\sqrt3},\tfrac{1}{\sqrt3}\Bigr)$, and therefore $\|\v r_\lambda\|_1=\frac{3-2\lambda}{\sqrt3}$. This gives the stabilizer threshold $\lambda_\star=\tfrac{3-\sqrt3}{2}$. However, the heat criterion is more restrictive. In the slice $E_0=3^{-1/2}$ one finds $R_\star(E_0)=\sqrt{\frac{5-2\sqrt3}{3}}$, which implies $\lambda_{\mathrm{crt}}=1-\sqrt{2-\sqrt3}<\lambda_\star$. Hence, there exists an interval $\lambda_{\mathrm{crt}}<\lambda<\lambda_\star$ in which the states remain nonstabilizer but are too mixed to be certified by the heat witness.
\begin{figure}[b]
    \centering
    \includegraphics{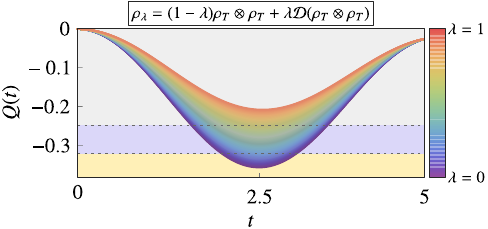}
    \caption{\emph{Magic detection in Tavis-Cummings model}. Heat exchange between two atoms and a single-mode optical cavity, which acts as an environment at inverse temperature $\beta = 1.5$, is shown as a function of time. The composite system is resonant with $\varepsilon = g = 1$. The two atoms are described by the Hamiltonian $H = Z_1 + Z_2$ and are prepared in the state $\rho_\lambda = (1 - \lambda)\rho_T^{\otimes 2} + \lambda\mathcal{D}(\rho_T^{\otimes 2})$. The yellow region represents nonstabilizer states that are detectable by our heat-based witness, while magical states that are undetectable by our witness are shown in the purple region. Stabilizer states lie in the gray region.}
    \label{F-Tavis-cummings}
\end{figure}
\begin{mybox}{Magic detection in the Tavis-Cummings model}
As an illustrative example of our heat-based witness, we consider two qubits coupled to a single bosonic mode $\ms{E}$ that plays the role of a thermal environment. The joint dynamics is governed by the resonant Tavis--Cummings Hamiltonian~\cite{Tavis1968}
\begin{align}
    H_{\ms{SE}} = \epsilon(Z_1 + Z_2) + \epsilon a_{\ms{E}}^\dagger a_{\ms{E}}
    + g \sum_{j=1}^2 \left(a_{\ms{E}} \sigma_j^\dagger + a_{\ms{E}}^\dagger \sigma_j\right),
\end{align}
where $a_{\ms{E}}$ is the annihilation operator of the environmental mode and $\sigma_j = \ket{0}\!\bra{1}_j$ is the lowering operator of qubit $j$. The environment is initially prepared in a thermal state at inverse temperature $\beta$, while the two-qubit system is prepared in the family
\begin{align}
    \rho_\lambda = (1-\lambda)\rho_T^{\otimes 2} + \lambda \mathcal{D}(\rho_T^{\otimes 2}).
\end{align}
Figure~\ref{F-Tavis-cummings} shows the heat exchanged with the environment as a function of time for different values of $\lambda$, for $\beta=1.5$ and $\epsilon=g=1$. As $\lambda$ increases, the heat curve moves upwards, capturing the loss of nonstabilizerness along the family. The dashed horizontal line marks the stabilizer heat threshold for this fixed energy. Heat values below this threshold certify nonstabilizerness, so the yellow region corresponds to states detected by the witness, while the purple region contains nonstabilizer states that are not detected. Stabilizer states lie in the gray region. 
\end{mybox}
\begin{figure*}
    \centering
    \includegraphics{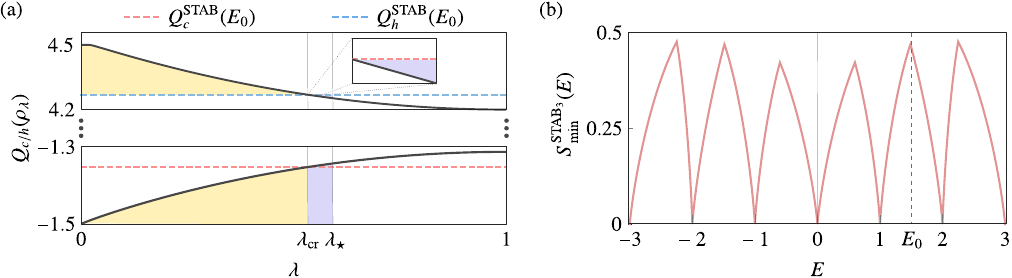}
    \caption{\emph{Magic detection via heat exchange for three-qubit system} (a) Heat exchange between a three-qubit system and an environment at inverse temperature $\beta = 0.01$ (dimensionless), shown as a function of $\lambda$. The system is prepared in the state $\rho_\lambda = (1-\lambda)\ketbra{\psi}+\lambda\,\mathcal{D}(\ketbra{\psi})$, with $\ket{\psi} \propto \sqrt{3}\ket{000}+e^{i\pi/4}\ket{111}$, and is characterized by the Hamiltonian $H = Z_1 + Z_2 + Z_3$. The state $\rho_\lambda$ is magical for $\lambda < \lambda_\star \approx 0.59$. The yellow area represents the range of $\lambda$ for which nonstabilizerness is detectable by our heat-based witness. The critical parameter $\lambda_{\text{crt}} \approx 0.53$ marks the point at which our witness ceases to detect it. Magical states that are undetectable by our witness are represented in the purple region. (b) Minimal entropy within the stabilizer polytope of three qubits as a function of energy. The dashed line represents the average energy of the family. A state $\rho$ is detectable by our witness when $S(\rho)<S_{\min}^{\textrm{STAB}{_3}}$.}
    \label{F:GHZ-magic}
\end{figure*}
These two examples show that optimality of the heat-based witness is not determined solely by whether the Hamiltonian has magic eigenstates. Rather, it depends on the geometry of the fixed-energy slice selected by the Hamiltonian and on how the noisy family approaches the stabilizer boundary. In the first example, the family reaches the stabilizer set through an entropy-minimizing point of the slice, whereas in the second it does not.

Beyond the single-qubit case, no closed Bloch-geometric formula is available in general. Nevertheless, the same fixed-energy principle remains valid: conditioned on the measured energy $E_0$, heat-based detection is governed by whether the state entropy lies below the minimum stabilizer entropy in the slice $\text{STAB}_n(E_0)$. We now illustrate this by taking a three-qubit example. Consider the noninteracting Hamiltonian $H = Z_1+Z_2+Z_3$ and the family $\rho_\lambda=(1-\lambda)\ketbra{\psi}+\lambda \mathcal D(\ketbra{\psi})$, where $|\psi\rangle = \tfrac{1}{2} (\sqrt{3}|000\rangle + e^{i\pi/4}|111\rangle)$ is a genuine multipartite entangled and magic state and $\mathcal D$ denotes dephasing in the eigenbasis of $H$ with $\mathcal{D}(\ketbra{\psi})$ being a stabilizer state. The resulting heat bounds are shown in Fig.~\hyperref[F:GHZ-magic]{\ref{F:GHZ-magic}(a)}. The heat-based witness detects nonstabilizerness up to the critical value $\lambda_{\mathrm{crt}}\approx 0.53$, whereas the true stabilizer threshold is $\lambda_\star\approx 0.59$. Hence, there exists a finite interval $\lambda_{\mathrm{crt}}<\lambda<\lambda_\star$ in which the states remain nonstabilizer but are already too mixed to be certified by heat exchange. In this sense, the witness is not tight along this family. At the same time, direct energy measurement is completely inconclusive: not only is the energy constant along the path, but the Hamiltonian itself has a vanishing stabilizer gap, since its ground space contains stabilizer states.

To understand why heat detection is nevertheless possible, one must look at the minimum stabilizer entropy on the corresponding fixed-energy slice. This quantity, shown in Fig.~\hyperref[F:GHZ-magic]{\ref{F:GHZ-magic}(b)}, is $S_{\min}^{\mathrm{STAB}_3}(E)$. At the relevant energy $E_0=3/2$, one finds $S_{\min}^{\mathrm{STAB}_3}(E_0)>0$. This means that every stabilizer state compatible with that energy is necessarily mixed, so stabilizer states on that slice cannot have arbitrarily low entropy. Since in a fixed-energy slice the free energy differs from the entropy only by an additive constant, this produces a free energy gap: any state with the same energy but entropy below $S_{\min}^{\mathrm{STAB}_3}(E_0)$ is detectable by the heat witness. 

More generally, Fig.~\hyperref[F:GHZ-magic]{\ref{F:GHZ-magic}(b)} identifies the energies for which heat-based detection can be nontrivial under the Hamiltonian $H$. For example, the $W$ state has an average energy $E_0=1$. In this slice, the stabilizer set already contains pure stabilizer states and therefore $S_{\min}^{\mathrm{STAB}_3}(1)=0$. Consequently, no state with energy $E_0=1$ can be certified by the heat witness, regardless of how magical it is, because a pure stabilizer state with the same energy has the same entropy and therefore the same free energy and optimal heat exchange. This shows that greater magic does not necessarily imply stronger heat-based detectability: the witness is controlled jointly by energy and entropy through the geometry of the fixed-energy slice. 

The example also clarifies why the witness is not tight. In the optimal single-qubit example of Fig.~\ref{F:heat-singlequbit-plot}, the noisy path reaches the stabilizer set through an entropy-minimizing point of the fixed-energy slice. Here this does not happen. Instead, the family intersects the stabilizer polytope away from such a point, so the entropy of $\rho_\lambda$ reaches the stabilizer entropy threshold before the true stabilizer boundary is reached. As a result, heat detection ceases at $\lambda_{\mathrm{crt}}<\lambda_\star$. For fixed-energy families, the detectability threshold is independent of $\beta$: although the heat values themselves depend on the bath temperature, the threshold $\lambda_{\mathrm{crt}}$ is determined by the entropy crossing $S(\rho_\lambda)=S_{\min}^{\mathrm{STAB}_n}(E_0)$, which contains no explicit $\beta$-dependence. 

Finally, this three-qubit example shows that heat exchange can reveal nonstabilizerness even when direct energy measurements are completely uninformative. At the same time, it also makes clear that tightness is exceptional rather than generic: it depends on the interplay between the Hamiltonian-induced energy slice and the direction along which the noisy family approaches the stabilizer boundary.

\section{Outlook}
We have introduced two operational routes to certifying magic using experimentally meaningful observables, bypassing the need for full tomography. The first is a direct energy witness, based on the stabilizer threshold—the minimum energy achievable by any stabilizer state for a given Hamiltonian. Any state with energy below this threshold is necessarily nonstabilizer. The difference between this threshold and the true ground-state energy defines the stabilizer gap, which quantifies how effective a Hamiltonian is for witnessing magic through energy measurements alone. The second approach is a heat-based witness, tailored for regimes where energy measurements are inconclusive. By considering optimal heat exchange between the system and a thermal ancilla, assisted by a cyclic quantum memory, this method incorporates both energetic and entropic features. As a result, it can distinguish states with identical average energy but different internal structure, making it inherently nonlinear and, in some scenarios, strictly more powerful than the direct energy witness. Pratically, the heat witness requires prior knowledge of the threshold $S_{\min}^{{\textrm{STAB}_n}}(E_0)$. Beyond the single-qubit case, this quantity is not known in closed form in general and must be obtained through a separate classical optimization over stabilizer states compatible with the measured energy. In this work, we evaluate for some examples. Thus, the experimental data required by the witness remain simple--the pair $(E_0, Q)$--but the classical preprocessing needed to determine the stabilizer heat window is model-dependent.

These results open several natural directions for future investigation. First, both witnesses are intrinsically binary: they detect the presence of nonstabilizerness but do not quantify it. A key next step is to relate the magnitude of the violation of the energy or heat bounds to established magic measures, thereby promoting these thermodynamic criteria to quantitative resource estimators. This direction is particularly compelling in light of the growing body of work connecting nonstabilizerness and thermodynamics~\cite{Koukoulekidis2022,junior2025tradingathermalitynonstabiliserness,trigueros2025nonstabilizernesserrorresiliencenoisy}.

Second, the witnessing power of a Hamiltonian is highly model-dependent. The stabilizer gap may vanish even for interacting, nonstabilizer Hamiltonians, as exemplified by the Heisenberg chain. This raises the question of whether one can systematically design—or even learn—Hamiltonians that maximize the stabilizer gap for a fixed system size or experimental constraint. Such constructions could serve as “optimal thermometers of magic" and may reveal a deeper connection to the geometry of the stabilizer polytope, for example, through its facet-defining operators.

Finally, it would be interesting to explore the robustness and scalability of our witnesses in realistic experimental settings, particularly in light of recent advances in measurement-efficient protocols inspired by classical shadow tomography. Instead of relying on full tomography, the nonstabilizerness of a quantum state can be inferred with sample complexity that scales only logarithmically in the number of observables of interest~\cite{varela2026predictingmagicmeasurements}. This suggests a promising route to implement our energy- and heat-based witnesses in a highly resource-efficient manner, where the required expectation values could be extracted from shadow data rather than dedicated measurements. Establishing a concrete integration between thermodynamic witnesses of nonstabilizerness and shadow-based protocols may therefore provide a scalable and experimentally viable framework for certifying magic in near-term quantum devices.

\label{S:outlook}

\begin{acknowledgments}
We acknowledge funding from Coordenação de Aperfeiçoamento de Pessoal de Nível Superior – Brasil (CAPES) – Finance Code 001, the Simons Foundation (Grant No. 1023171, R.C.), the Brazilian National Council for Scientific and Technological Development (CNPq, Grants No.403181/2024-0, 308065/2022-0, and 301687/2025-0), the National Institute of Science and Technology for Applied Quantum Computing through (CNPq grant 408884/2024-0), the Financiadora de Estudos e Projetos (Grant No. 1699/24 IIF-FINEP), a guest professorship from the Otto M\o nsted Foundation, the Danish National Research Foundation grant bigQ (DNRF 142) and EU Horizon Europe (QSNP, grant no.
10111404). LCC acknowledges support from FAPEG through grant 202510267001843, and FAPESP through grant 2025/23726-4.
\end{acknowledgments}

\emph{Data availability--} The code is publicly available on GitHub at~\href{https://github.com/AdeOliveiraJunior}{Thermodynamic witnesses of magic}

\bibliography{2-references}

\appendix
\onecolumngrid

\section{Heat-based witness for nonstabilizerness}
\label{Sec:app-optimal-heat-exchange}

In this appendix, we recall the main arguments leading to the expression for the optimal heat exchange. Rather than reproducing the full technical proof, our aim is to convey the key conceptual steps and physical insights behind the result. A complete and rigorous derivation can be found in Ref.~\cite{de2025heat}. We then leverage these results to construct a heat-based witness for nonstabilizerness, and conclude by proving Theorem~\ref{Thm:heat-detectability} and Corollary~\ref{Cor:optimality}.

\subsection{The optimal heat exchange problem}

We consider a system $\S$ with Hamiltonian $H_\S$, an environment $\E$ prepared in the Gibbs state $\gamma_\E(\beta)$ at inverse temperature $\beta$, and a quantum memory $\M$. The global evolution is governed by an energy-preserving unitary $U$ satisfying $[U,H_\S+H_\M+H_\E]=0$, which maps the initial uncorrelated state $\rho_\S\otimes\rho_\M\otimes\gamma_\E$ to a correlated state $\eta_{\S\M\E}$. The only constraints are energy conservation and the catalytic condition that the memory returns to its initial state $\rho_\M$ at the protocol's end, i.e., $\eta_\M=\rho_\M$ where $\eta_\M:=\tr_{\S\E}(\eta_{\S\M\E})$. Since no external work is supplied, the only energetic exchange is heat between $\S$ and $\E$, defined as $Q=\tr[H_\E(\eta_\E-\gamma_\E)]$, where $\eta_\X$ denotes the final state of subsystem $\X\in\{\S,\E,\M\}$.

Our goal is to determine the optimal heat exchanged with the environment, specifically the heat released $Q_c$ and the heat absorbed $Q_h$, for a fixed initial state $\rho_\S$. Formally,
\begin{align}
\label{app-Eq:app-optimal-heat}
\begin{split}
    Q_{c/h}(\rho_{\ms S}) :=   \underset{H_{\ms E}, H_{\ms M}, U, \rho_{\ms{M}}}{\min/\max} \:\: &\tr[H_{\ms{E}}(\eta_{\ms{E}} - \gamma_{\ms{E}})], \\ 
\textrm{s.t.}  \quad\quad &[U, H_{\ms{S}} +H_{\ms{M}} + H_{\ms{E}}] = 0, \\   &\eta_{\ms M}=\rho_{\ms{M}}. 
\end{split}
\end{align}
Although Eq.~\eqref{app-Eq:app-optimal-heat} appears complex, this optimization admits a simple reduction. Every feasible transformation must satisfy the free-energy inequality $F_\beta(\rho_S) \ge F_\beta(\eta_S)$~\cite{brandao2015second,shiraishi2025}. Conversely, within catalytic coherent thermal operations, any target state $\eta_S$ fulfilling this constraint can be approximated arbitrarily well with a sufficiently large catalyst (see~\cite{Bartosik2023, de2025heat} for the constructive proof). Therefore, the optimization in \eqref{Eq:app-optimal-heat} is equivalent to the convex program:
\begin{equation}
\begin{aligned}
\label{Eq:app-optimal-heat-2}
    Q_{c/h}(\rho_{\ms S}) := \underset{\eta_{\ms S}}{\min\!/\!\max} \:\: &\tr[H_{\ms{S}}(\rho_{\ms S}-\eta_{\ms S})],\\
&\hspace{-0.95cm}\textrm{s.t.} \:\:\:\:\:\:\:\:  F_{\beta}(\rho_{\ms S}) \geq F_{\beta}(\eta_{\ms S}).
\end{aligned}
\end{equation}
This reformulation isolates the single information-theoretic constraint that governs all achievable heat exchanges. The free-energy constraint is saturated at the optimum. Introducing the thermal states of the system parameterized by an effective inverse temperature $x$:
\begin{equation}
\gamma_\S(x):=\frac{e^{-xH_\S}}{\tr(e^{-xH_\S})}.
\end{equation}
The optimal final state must then be of the form $\eta_\S=\gamma(x)$, where $x$ is determined implicitly by $F_\beta(\rho_\S)=F_\beta[\gamma(x)]$. The behavior of the root equation follows from
\begin{equation}
\frac{d}{dx}F_\beta[\gamma(x)] = \qty(\frac{x}{\beta}-1)\mathrm{var}_{\gamma(x)}(H_\S),
\label{Eq:freeenergyder}
\end{equation}
which shows that $F_\beta[\gamma(x)]$ decreases for $x<\beta$, attains a unique minimum at $x=\beta$, and increases for $x>\beta$ (see Fig.~\hyperref[F-app-free-energy-heat]{\ref{F-app-free-energy-heat}(a)} for an illustration). This guarantees at most two real roots, which are two inverse temperatures $\beta_c\le \beta_h$ associated with the cooling and heating solutions. The optimal heats follow immediately:
\begin{equation}\label{Eq:app-Q-opt}
Q_{c/h}(\rho_\S)=E(\rho_\S)-E[\gamma_\S(\beta_{c/h})].
\end{equation}

Equipped with the optimal-heat formula, we construct a witness for the hypothesis that $\rho_S$ belongs to a specific set $\mathcal{S}$ of states, which we will assume to be convex (e.g., separable states, incoherent states, or stabilizers).
However, the heat bound derived from this approach depends on the state's energy. Since on many occasions the average energy $E_0 := E(\rho_\S)$ can be measured prior to the protocol, we condition on this value and refine the hypothesis to the set $\mathcal{S}(E_0) := \{\sigma \in \mathcal{S} : E(\sigma) = E_0\}$. We then maximize the free energy over $\mathcal{S}(E_0)$. Because $F_\beta(\sigma) = E_0 - \beta^{-1} S(\sigma)$ for $\sigma \in \mathcal{S}(E_0)$, this reduces to minimizing the entropy
\begin{equation}\label{Eq:app-free-energy-conditioned}
    F^{\star}_\beta(\mathcal{S}|E_0):=\max_{\sigma\in \mathcal{S}(E_0)}F_\beta(\sigma) = E_0 -\frac{1}{\beta}S_{\min}^{\mathcal{S}|E_0},
\end{equation}
where $S_{\min}^{\mathcal S|E_0}:=\min_{\sigma\in\mathcal S(E_0)} S(\sigma)$. Then, solving $F_\beta[\gamma(x^\star_{c/h})] =F^{\star}_\beta(\mathcal{S}|E_0)$ gives the effective temperatures $x^\star_{c/h}$. Every state $\sigma \in \mathcal{S}(E_0)$ then satisfies
\begin{equation}\label{Eq:app-heat-stabilizer-bound}
Q(\mathcal{S}|E_0)\in\big[E_0-E(x^\star_c),\ E_0-E(x^\star_h)\big],
\end{equation}
so an observed heat outside this interval refutes the hypothesis that $\rho_S \in \mathcal{S}(E_0)$, and hence that $\rho_S \in \mathcal{S}$.

\subsection{Single-qubit witness}\label{App:single-qubit-witness}

We now apply the framework developed in the previous sections to the single-qubit case, construct a heat-based witness for nonstabilizerness, and prove Theorem~\ref{Thm:heat-detectability} and Corollary~\ref{Cor:optimality}. Let $\mathcal{S} = \mathrm{STAB}_1$ and consider a system $\S$ described by a traceless, unit-gap Hamiltonian $H_\S= \v{h} \cdot \v{\sigma}$ with spectrum $\{\pm 1\}$.

The thermal Gibbs state at inverse temperature $x$ is given by $\gamma(x) = \operatorname{diag}\left(\tfrac{1}{1+e^{2x}},\tfrac{e^{2x}}{1+e^{2x}}\right)$. For this state, the relevant thermodynamic quantities are:
\begin{align}
E[\gamma_\S(x)] &= -\tanh x, \\
S[\gamma_\S(x)] &= \log(2\cosh x) - x\tanh x.
\end{align}
The non-equilibrium free energy of this state is then $F_\beta[\gamma_\S(x)] = -\frac{\beta-x}{\beta}\tanh x-\frac{1}{\beta}\log(2\cosh x)$. 

Given the measured average energy $E_0:=E(\rho_\S)$, we restrict to the energy-constrained subset $\mathrm{STAB}_1(E_0):=\{\sigma \in \mathrm{STAB}_1: E(\sigma) = E_0\}$. Consequently, our figure of merit [Eq.~\ref{Eq:app-free-energy-conditioned}] is
\begin{equation}\label{Eq:app-free-energy-constrained-stab}
    F^{\star}_\beta(\stab|E_0):=\max_{\sigma\in \stab(E_0)}F_\beta(\sigma) = E_0 -\frac{1}{\beta}S_{\min}^{\stab|E_0}.
\end{equation}
To compute this minimum entropy, we note that the qubit entropy is a strictly decreasing function of the Bloch radius for $\|\v r\|_2 >0,$
\begin{equation}
    S(\rho) = H_2\qty(\frac{1+\|\v r\|_2}{2}) \quad \text{where} \quad H_2(p):=-p\log p - (1-p)\log (1-p).
\end{equation}
Hence, minimizing $S(\rho)$ at fixed energy $E_0$ is equivalent to maximizing $\|\v r\|_2$  over $\{\v r: \|\v r\|_1\leq 1, \v h \cdot \v r = E_0\}$. Since $\|\v r\|_2$ is a convex function and the feasible region is a convex polytope, the maximum is attained at an extreme point. For the octahedron intersected with the plane $\v h \cdot \v r = E_0$, the optimizers lie on edges where $\|\v r\|_1 =1$ and exactly two components of $\v{r}$ are nonzero.
 \begin{figure*} 
    \centering
    \includegraphics{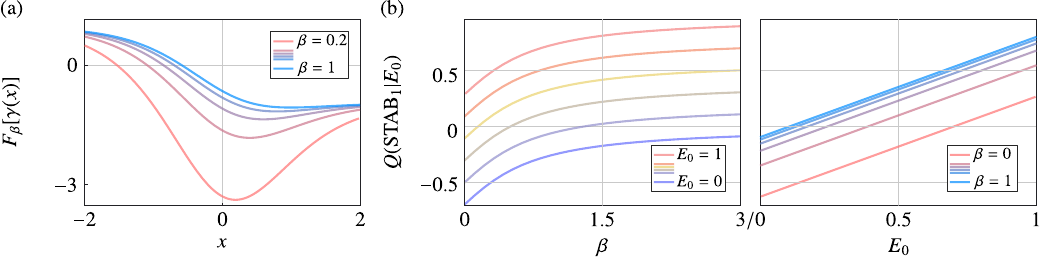}
    \caption{\emph{Free energy \& heat-based witness}. (a) Nonequilibrium free energy as a function of the effective temperature $x$ for different values of $\beta$. (b) Heat-based witness as a function of $\beta$ for different values of $E_0$ and vice-versa}
    \label{F-app-free-energy-heat}
\end{figure*}

Without loss of generality, consider an edge corresponding to axes $(i,j)$. Set $r_k =0$ for $k\neq i, j$, and choose signs $s_i, s_j \in \{\pm 1\}$, such that the edge can satisfy $\v h \cdot \v r = E_0$. Let $a:=|h_i|, b:=|h_j|$ and parametrize $r_i = s_i t$, $r_j = s_j(1-t)$, $r_k=0$, with $s_i h_i > 0$ and $s_j h_j <0$. Then the energy constraint becomes
\begin{equation}
    at - b (1-t) = E_0 \Longrightarrow t = \frac{E_0+b}{a+b}, 
\end{equation}
which is feasible if and only if $-b \leq E_0 \leq a$. On such a feasible edge, the Bloch radius is
\begin{equation}\label{Eq:app-Bloch-radius}
    \|\v r\|_2 = \sqrt{t^2+(1-t)^2} = \frac{\sqrt{(a-E_0)^2+(E_0+b)^2}}{a+b}.
\end{equation}
Maximizing Eq.~\eqref{Eq:app-Bloch-radius} over all feasible pairs $(i,j)$ gives $\|\v r\|_2^{\max}(E_0)$ and hence 
\begin{equation}\label{Eq:app-S-stab-min}
    S_{\min}^{\stab|E_0}= H_2\qty(\frac{1+\|\v r\|_2^{\max}(E_0)}{2}).
\end{equation}
Note that if we choose $H_\S$ with $|h_x|=|h_y|=|h_z|=\tfrac{1}{\sqrt{3}}$ and ensure $E_0=0$ (e.g. by taking $\v{h}\perp \v r_T$ for the $T$ state), then $a=b$ and \eqref{Eq:app-Bloch-radius} gives $t=\tfrac12$.
Equation~\eqref{Eq:app-Bloch-radius} then yields $\|r\|_2^{\max}(0)=\tfrac{1}{\sqrt{2}}$ and
\begin{equation}
S_{\min}^{\stab|E_0}=H_2\!\left(\frac{1+\tfrac{1}{\sqrt{2}}}{2}\right).
\label{Eq:Smin-balanced}
\end{equation}
Finally, we substitute the maximum stabilizer free energy $F_\beta^{\star}(\mathrm{STAB}_1|E_0)$ from Eq.~\eqref{Eq:app-free-energy-constrained-stab} into the equation $F_\beta[\gamma(x^\star)] = F_\beta^{\star}(\mathrm{STAB}_1|E_0)$ to determine $x^\star$, and then insert into Eq.~\eqref{Eq:app-heat-stabilizer-bound} to obtain the heat-based stabilizer bound:
\begin{equation}\label{Eq:app-heat}
    Q(\stab|E_0)\in\qty[Q_{\min}^{\stab|E_0},Q_{\max}^{\stab|E_0}]:=\qty[E_0 - E(x^\star_c), E_0-E(x^\star_h)].
\end{equation}
In Fig.~\hyperref[F-app-free-energy-heat]{\ref{F-app-free-energy-heat}(b)}, we illustrate the behavior of the heat-based witness as a function of both $\beta$ and $E_0$.

Assume $\stab(E_{\v r})\neq\emptyset$, so that the constrained stabilizer threshold is well-defined. By the construction above, the heat-based stabilizer witness detects $\rho$ exactly when the optimal heat attainable from $\rho$ lies outside the stabilizer window associated with the same energy slice. By the fixed-energy free-energy reduction in Eq.~\eqref{Eq:app-free-energy-constrained-stab} and the heat bounds in Eq.~\eqref{Eq:app-heat}, this is equivalent to $F_\beta(\rho) > F_\beta^\star(\stab|E_{\v r})$. Since $E(\rho)=E_{\v r}$ and every $\sigma\in\stab(E_{\v r})$ has the same energy $E_{\v r}$, we have $F_\beta(\rho)=E_{\v r}-\frac{1}{\beta}S(\rho)$ and $ F_\beta^\star(\stab|E_{\v r})=E_{\v r}-\frac{1}{\beta}S_{\min}^{\stab|E_{\v r}}$. Therefore, the witness detects $\rho$ if and only if $S(\rho)<S_{\min}^{\stab|E_{\v r}}$. Now, for a qubit state with Bloch vector $\v u$, the entropy is $ S[\rho(\v u)] = H_2\qty(\frac{1+\|\v u\|_2}{2})$, which is a strictly decreasing function of $\|\v u\|_2$. Hence, minimizing the entropy over $\stab(E_{\v r})$ is equivalent to maximizing the Bloch radius over the same slice, as discussed above Eq.~\eqref{Eq:app-S-stab-min}. It follows that
\begin{equation}
    S(\rho)<S_{\min}^{\stab|E_{\v r}} \iff \|\v r\|_2 > \max\qty{\|\v s\|_2:\|\v s\|_1\leq 1,\ \v h\cdot \v s = E_{\v r}}.
\end{equation}
By definition, the right-hand side is precisely $R_\star(E_{\v r})$, and thus $\rho$ is detected if and only if $\|\v r\|_2 > R_\star(E_{\v r})$, which proves Theorem~\ref{Thm:heat-detectability}.

For the family $\{\rho_\lambda\}_{\lambda\in[0,1]}$ we have, by assumption, $\tr(\rho_\lambda H)=E_0$ for all $\lambda$. Therefore, Theorem~\ref{Thm:heat-detectability} applies on the same fixed-energy slice for every $\lambda$, and gives $\rho_\lambda$ is detected if and only if $S(\rho_\lambda)<S_{\min}^{\stab|E_0}.$ Equivalently,
\begin{equation}
    \rho_\lambda \text{ is detected }
    \iff
    \|r_\lambda\|_2 > \max\qty{\|\v s\|_2:\|\v s\|_1\leq 1,\ \v h\cdot \v s = E_0}.
\end{equation}
Since $S(\rho_\lambda)$ is strictly increasing in $\lambda$, the set of detected states is an initial interval in $\lambda$. Hence, the heat-based witness is tight on this family if and only if the detection threshold coincides with the true stabilizer threshold $\lambda_\star$, namely if and only if $S(\rho_{\lambda_\star}) = S_{\min}^{\stab|E_0}$. Using again that, for qubits, the entropy is a strictly decreasing function of the Bloch radius, this is equivalent to
\begin{equation}
    \|\v r_{\lambda_\star}\|_2
    =
    \max\qty{\|\v s\|_2:\|\v s\|_1\leq 1,\ \v h\cdot \v s = E_0},
\end{equation}
which proves Corollary~\ref{Cor:optimality}.

\section{Stabilizer ground state energy\label{app:stabgroundenergy}}
In this Appendix, we review the results of \cite{sun2025stabilizerground}, which states that $E_\mathrm{STAB}(H)$ can be obtained by performing a restricted optimization instead of computing the energy of every vertex of the stabilizer polytope.

First, we note that, given a maximal stabilizer group $\mathbf Q \subseteq \mathcal P_n$, it is easy to compute its energy. Define the stabilizer character of $\mathbf Q $ as, given $P \in \mathcal P_n$ :
\begin{equation}
    \chi_{\mathbf S}(P):=
    \begin{cases}
        +1, & P\in\mathbf S,\\
        -1, & -P\in\mathbf S,\\
        0,  & \text{otherwise},
    \end{cases}
\end{equation}
such that $\langle \mathbf Q | H | \mathbf Q \rangle = \chi_\mathbf Q(P)$. Then:
\begin{align}
\bra{\mathbf Q}H \ket{\mathbf Q} &=  - \sum_{P \in  \mathbf P(H)}w_P \;\chi_\mathbf Q(P)\;. \label{Eq:stabilizerenergy}
\end{align}
As a consequence, we have the following result:
\begin{lem} [Theorem 1 of \cite{sun2025stabilizerground}]
    The stabilizer ground state energy problem can be expressed as:
    \begin{equation}
        E_\mathrm{STAB}(H) = \min_{\rho \in \mathrm{STAB}_n} \tr(\rho H) = -\max_{\mathbf Q \in C_{\max}[\mathbf P(H)]}\sum_{P \in\mathbf Q}w_P\;,
    \end{equation}
    where $C_\mathrm{max}[\mathbf P(H)]$ is the set of (maximal) closed commuting subsets of $\mathbf P(H)$:
    \begin{equation}
        C_{\max}[\mathbf P(H)] = \{\mathbf Q \subseteq \mathbf P(H) | \mathbf Q =\langle \mathbf Q \rangle \cap \mathbf P(H), \; -1 \notin \langle \mathbf Q \rangle\;\} \cap \mathbf P_{\max }(H)\;,
    \end{equation}
    where $\mathbf P_{\max}(H) \equiv \{\mathbf Q \subseteq \mathbf P(H)| \nexists \mathbf Q^\prime \subseteq \mathbf P(H): \mathbf Q \subset \mathbf Q^\prime\}$ are the subsets of $\mathbf P(H)$ which are maximal in cardinality.
    \label{lem:stabilizerenergy}
\end{lem}
\begin{proof}
    We start by expressing the reducing the minimization to the maximal stabilizer groups:
    \begin{equation}
        E_\mathrm{STAB}(H) = \min_{\rho \in \mathrm{STAB}_n} \tr(\rho H) =\min_{\mathbf S \in \mathrm{ext}(\mathrm{STAB}_n)} \bra{\mathbf S} H \ket{\mathbf S}\;.
    \end{equation}
    Then, note that Eq. (\ref{Eq:stabilizerenergy}) can be expressed as:
    \begin{equation}
        \bra{\mathbf S} H \ket{\mathbf  S} = - \sum_{P \in \mathbf S \cap \mathbf P(H)} w_P + \sum_{P \in \mathbf S \cap -\mathbf P(H)}w_P\;,
        \label{Eq:groupstabenergy}
    \end{equation}
    where $-\mathbf P(H)$ is the set of Paulis of $H$ with signs interchanged, where, of course, $\mathbf P(H) \cap -\mathbf P(H) = \varnothing$. Due to the positivity of the coefficients, energy of $\mathbf S$ is minimized if $|\mathbf S \cap -\mathbf P(H)|=0$, on which we restrict the optimization. Furthermore, note that given two $\mathbf S, \mathbf S^\prime \in \mathrm{ext}(\mathrm{STAB}_n)$ such that $\mathbf S \cap \mathbf P(H) = \mathbf S^\prime \cap \mathbf P(H)$, their energies are the same, since then $\bra{\mathbf S} H \ket{\mathbf S} = \bra{\mathbf S^\prime} H \ket{\mathbf S^\prime}$.
    
    Note that, given $\mathbf S \in \mathrm{ext}(\mathrm{STAB}_n)$, the corresponding intersection $\mathbf Q =   \mathbf S \cap \mathbf P(H_S)$ must generate a stabilizer group, since (1) $\mathbf Q$ is composed of commuting elements, and (2) $-I \notin \langle \mathbf Q \rangle$, since $\langle \mathbf Q\rangle \subseteq \mathbf S $ is a subgroup. Then, we show that optimization can be restricted to the set:
    \begin{equation}
       \mathrm{STAB}(H) \equiv \{\mathbf S \in \mathrm{ext}(\mathrm{STAB}_n)| \langle \mathbf S \cap \mathbf P(H)\rangle \mathrm{\;is\;stabilizer}\}\;.
    \end{equation}
    Consider now the set:
    \begin{equation}
        C[\mathbf P(H)] = \{\mathbf Q \subseteq \mathbf P(H)|\mathbf Q = \langle \mathbf Q \rangle \cap \mathbf P(H), \; -1 \notin \langle \mathbf Q \rangle\}\;,
    \end{equation}
    We claim that there is a surjective map $\mathrm{STAB}(H) \to C[\mathbf P(H)]$,
    constructed as follows:
    \begin{enumerate}
        \item Consider $\mathbf S \in \mathrm{STAB}(H)$, and let $ \mathbf S \mapsto \mathbf Q =\mathbf S \cap \mathbf P(H) \subseteq \mathbf P(H)$ be the action of the map. Note that: (1) since  $\langle \mathbf S \cap \mathbf P(H)\rangle \subseteq \mathbf S$, it follows that $\langle \mathbf Q \rangle \cap \mathbf P(H) =  \mathbf Q$. (2) Since $\langle \mathbf Q \rangle \subseteq \mathbf S$, $-1 \notin \langle \mathbf Q\rangle$. Hence, we have $\mathbf Q \in C[\mathbf P(H)]$.
        \item It is energy-preserving, due to the property that if two stabilizer groups have the same intersection with $\mathbf P(H)$, they have the same energy. Hence,
        \begin{equation}
            E_\mathrm{STAB}(H) = \min_{\mathbf S \in \mathrm{STAB}(H)} \bra{\mathbf S} H \ket{\mathbf S} = \min_{\mathbf Q \in C[\mathbf P(H)]} \left(-\sum_{P \in \mathbf Q}w_P\right) = -\max_{\mathbf Q \in C[\mathbf P(H)]} \sum_{P \in \mathbf Q}w_P\;,
        \end{equation}
        where we have used Eq. (\ref{Eq:groupstabenergy}). 
    \end{enumerate}
    To finish the proof, we need to show maximality. But this follows from the positivity of the coefficients in the expression above: Let $\mathbf Q_1, \mathbf Q_2 \in C[\mathbf P(H)]$, such that $\mathbf Q_1 \subseteq \mathbf Q_2$. Then, it follows that:
    \begin{equation}
        \sum_{P \in \mathbf Q_1}w_P \leq \sum_{P \in \mathbf Q_2}w_P\;,
    \end{equation}
    since all $w_P$ are positive. Therefore, optimization can be restricted to $C_{\max}[\mathbf P(H)]$.
\end{proof}

Hence, a straightforward algorithm is avaliable to estimate the stabilizer ground state energy in $|C[\mathbf P(H)]| \times |\mathbf P(H)|$ time: We run through all the stabilizer groups the Hamiltonian generates, and evaluate the energy of each of then, and we take the minimum. In fact, it was also shown that there is a constant $c=O(1)$ such that $| C[\mathbf P(H)]| = O(\exp(c n \ln n))$. This is polynomially better (modulo $log$ factors) than the $|\mathrm{ext}(\mathrm{STAB}_n)|=2^{O(n^2)}$ extremal points of the stabilizer polytope \cite{aaronson2004improved} one must go through in direct optimization. However, this is a worst-case scenario: For $k-$local Hamiltonians in one dimension, they show a $O(n \exp(ck \log k))$ runtime algorithm to compute the minimum, which is linear time for $k=O(1)$. 

\section{Stabilizer gap with perturbations: Degenerate case \label{app:degeneratestabilizergap}}
We show here a similar version of Lemma \ref{lem:perturb-stabilizer}, when we allow the perturbed stabilizer Hamiltonian to have a degenerate ground space. Given $\mathbf S \subseteq \mathcal P_n$ a stabilizer group, denote a corresponding stabilizer Hamiltonian defined by a generating set as Eq. (\ref{Eq:stabilizer-Hamiltonian}) by $H_\mathbf S$. First, we show that under certain natural conditions, there is a quantum state that has energy lower than the one provided by a state in the code of the unperturbed Hamiltonian:
\begin{lem}
    Let $H =H_\mathbf Q+ V$ be a Hamiltonian, with $\mathbf Q \subseteq \mathcal P_n$ be a stabilizer group. Denote $V_\perp = -\sum_{P \in \mathbf P(V)\backslash C(\mathbf Q)}w_P P\neq 0$ as the part of the pertubation $V$ whose Paulis are not in the centralizer of $\mathbf Q$. Then, if:
    \begin{itemize}
        \item There is quantum state $|\psi^\prime_\mathbf Q\rangle \in V_\mathbf Q$ such that $\langle V_\perp^2\rangle_{\psi^\prime_\mathbf Q} \neq 0$,
        \item The perturbation is weak enough, $\sum_{P \in \mathbf P(V) \backslash \mathbf Q} w_P < \Delta_\mathrm{gap}-2\sum_{P \in \mathbf P(V) \cap \mathbf Q}w_P$, where $\mathrm{\Delta}_\mathrm{gap} = 2\min_{P \in \mathrm{gen}(\mathbf Q)} w_P$ is the energy gap of $H_\mathbf Q$.
    \end{itemize}
    Then, there is a $\lambda_\mathrm{max} >0 $ such that the quantum state defined through $|\psi^\prime_\mathbf Q(\lambda)\rangle \propto |\psi_\mathbf Q^\prime\rangle-\lambda V_\perp |\psi^\prime_\mathbf Q\rangle$ such that $\Delta E(\lambda) =\langle H \rangle_{\psi^\prime_\mathbf Q}-\langle \psi^\prime_\mathbf Q(\lambda)|H|\psi^\prime_\mathbf Q(\lambda)\rangle $ is strictly positive and monotonically increasing for all $0 < \lambda < \lambda_\mathrm{max}$.
    \label{lem:degperturbstabilizer}
\end{lem}
\begin{proof}
    We first start by decomposing the Hamiltonian $H$ into:
\begin{equation}
    H =H_\mathbf{Q} + V = H_{\mathbf Q} - \sum_{P \in \mathbf P(V)\cap C(\mathbf Q)}w_P P - \sum_{P \in \mathbf P(V)\backslash C(\mathbf Q)}w_P P\;,
\end{equation}
where we denoted $C(\mathbf Q) \equiv \{\P \in \mathcal P_n \;|\; [P, Q] =0, \; \forall Q \in \mathbf Q\}$ as the Pauli centralizer of $\mathbf Q$. This separation is relevant,  since it splits the perturbation into the terms that preserve the ground state subspace and another which lifts it. Denoting $\mathbf P_\parallel(\mathbf Q) = \mathbf P(V) \cap C(\mathbf Q)$ and $\mathbf P_\perp(V) = \mathbf P(V) \backslash C(\mathbf Q)$ with:
\begin{equation}
     V_{\parallel} \equiv - \sum_{P \in \mathbf P_\parallel(V)} w_P P\quad ;\quad V_\perp \equiv -\sum_{P \in \mathbf P_\perp(V)}w_P P\;.
\end{equation}
 We know that $H_\mathbf Q$ stabilizes a code, with a basis $V_\mathbf Q = \mathrm{span}_\mathbb C \left \{|\psi_\mathbf Q^1 \rangle , |\psi^2_\mathbf Q \rangle, \cdots, |\psi^{2^k}_\mathbf Q \rangle \right \}$, where $k=n-\mathrm{rank}(\mathbf Q)$, associated to a projector:
 \begin{equation}
     \Pi_\mathbf Q= \frac{1}{|\mathbf Q|}\sum_{P \in \mathbf Q} P\;.
 \end{equation}
 We are going to consider the subspace decomposition $\mathcal H = V_\mathbf Q \oplus V_\mathbf Q^\perp$. We can now move to the construction of the corresponding variational state. Let $|\psi^\prime_\mathbf Q\rangle  =\sum_{a=1}^{2^k} c_a |\psi^a_\mathbf Q \rangle \in V_\mathbf Q$ be the code state satisfying $\langle V^2_\perp\rangle_{\psi^\prime_\mathbf Q} \neq 0$, by assumption. More specifically, we have:
\begin{equation}
\langle V^2_\perp\rangle_{\psi^\prime_\mathbf Q} = \sum_{\substack{P, P^\prime \in \mathbf P_-(V)\\ PP^\prime \in C(\mathbf Q)}} w_P w_{P^\prime} \langle PP^\prime \rangle_{\psi^\prime_\mathbf Q} >0\;.
\end{equation}
One can construct the following variational state:
\begin{equation}
    |\psi_\mathbf Q^\prime(\lambda)\rangle  \equiv \frac{|\psi^\prime _\mathbf Q\rangle-\lambda V_\perp |\psi_\mathbf Q^\prime\rangle }{\sqrt{1+ \lambda^2\langle V^2_\perp\rangle_{\psi^\prime_\mathbf Q} }}\;.
    \label{eq:ansatzstate}
\end{equation}
Note that $V_\perp |\psi_\mathbf Q^\prime \rangle \in V_\mathbf Q^\perp$, since, given $P \notin C(\mathbf Q)$, $P|\psi_\mathbf Q\rangle \notin V_\mathbf Q$, due to the fact that $P$ must anticommute with at least one generator. In the following, we need to guarantee that the state $|\xi \rangle \equiv V_\perp |\psi^\prime_\mathbf Q\rangle/\sqrt{\langle V_\perp^2 \rangle_{\psi^\prime_\mathbf Q}}$ has a higher energy than $|\psi^\prime_\mathbf Q\rangle$, that is:
\begin{equation}
    \frac{\langle V_\perp H V_\perp \rangle_{\psi^\prime_\mathbf Q}}{\langle V_\perp^2\rangle_{\psi^\prime_\mathbf Q}} > \langle H \rangle_{\psi^\prime_\mathbf Q}\;.
    \label{Eq:variationalstabilityproof}
\end{equation}
Let us show that this is indeed the case if the perturbation is weak enough. Denoting $E_0=-\sum_{G \in \mathrm{gen}(\mathbf Q)}w_G$ as the ground space energy, and $\Delta_\mathrm{gap}$ as the energy gap, we know that:
\begin{equation}
    \frac{\langle V_\perp H_\mathbf Q V_\perp\rangle_{\psi^\prime_\mathbf Q}}{\langle V_\perp^2\rangle_{\psi^\prime_\mathbf Q}} -E_0\geq \Delta_\mathrm{gap}\;,
    \label{Eq:gaplowerbound}
\end{equation}
since $V_\perp$ lifts the ground state space to a excited state of $H_\mathbf Q$. The subtle term is the interactions. First, note that we can write the code-preserving interactions as:
\begin{equation}
    V_\parallel = -\sum_{P \in \mathbf P_\parallel(V)}w_P P = -\sum_{P \in \mathbf P(V) \cap \mathcal L(\mathbf Q)} w_P P - \sum_{P \in \mathbf P(V) \cap \mathbf Q} w_P P\;,
\end{equation}
where $\mathcal L(\mathbf Q) = C(\mathbf Q) \backslash \mathbf Q$ are called the logical Paulis of $\mathbf Q$,and they act non-trvially in the code subspace, and thus this first term vanishes when taking the expectation value with $\psi^\prime_\mathbf Q$. Therefore:
\begin{equation}
    \langle V \rangle_{\psi^\prime_\mathbf Q} =  \underbrace{\langle V_\perp \rangle_{\psi^\prime_\mathbf Q}}_{=0} + \langle V_\parallel\rangle_{\psi^\prime_\mathbf Q} = - \sum_{P \in \mathbf P(V)\cap \mathbf Q} w_P.
\end{equation}
Then, by Eq. (\ref{Eq:gaplowerbound}), we have:
\begin{equation}
    \frac{\langle V_\perp H V_\perp \rangle_{\psi^\prime_\mathbf Q}}{\langle V_\perp^2\rangle_{\psi^\prime_\mathbf Q}} -\langle H \rangle_{\psi^\prime_\mathbf Q} \geq \Delta_\mathrm{gap} + \langle V\rangle_\xi -\langle V\rangle_{\psi^\prime_\mathbf Q}\;,
\end{equation}
and since $\langle V\rangle_\xi \geq -\|V\|\geq - \sum_{P \in \mathbf P(V)}|w_P|$, we have:
\begin{equation}
    \frac{\langle V_\perp H V_\perp \rangle_{\psi^\prime_\mathbf Q}}{\langle V_\perp^2\rangle_{\psi^\prime_\mathbf Q}} -\langle H \rangle_{\psi^\prime_\mathbf Q}  \geq \underbrace{\Delta_\mathrm{gap}-\sum_{P \in \mathbf P(V) \cap \mathbf Q} w_P - \sum_{P \in \mathbf P(V)}w_P}_{=\Delta_\mathrm{gap}-2\sum_{P \in \mathbf P(V) \cap \mathbf Q} w_P - \sum_{P \in \mathbf P(V)\backslash \mathbf Q}w_P  }\;,
    \label{eq:fluctuationenergydiff}
\end{equation}
explicitly positive under the assumption of the Lemma.

\noindent 
Now we will show that that there exists $\lambda > 0$ such that $\psi_\mathbf Q^\prime (\lambda)$ will hold a lower variational energy of $H$ compared to $\psi_\mathbf Q^\prime$. Computing the variational energy:
\begin{align}
    E(\lambda) &= \langle \psi_\mathbf Q^\prime (\lambda) |H| \psi_\mathbf Q^\prime(\lambda) \rangle =  \frac{1}{1+ \lambda^2\langle V_\perp^2\rangle_{\psi^\prime_\mathbf Q}} \left[\langle H \rangle_{\psi^\prime_\mathbf Q}- 2\lambda \mathrm{Re} \langle V_\perp H \rangle_{\psi_\mathbf Q^\prime}  + {\lambda^2} \langle V_\perp H V_\perp\rangle_{\psi^\prime_\mathbf Q}\right]\\
    &= \frac{1}{1+ \lambda^2 \langle V_\perp^2\rangle} \left[\langle H \rangle_{\psi^\prime_\mathbf Q}- 2\lambda \langle V_\perp^2 \rangle_{\psi_\mathbf Q^\prime}  + {\lambda^2} \langle V_\perp H V_\perp\rangle_{\psi^\prime_\mathbf Q}\right]\;,
\end{align}
where we have used the fact that:
\begin{equation}
     \langle V_\perp H \rangle_{\psi_\mathbf Q^\prime} =\langle V_\perp(1-\Pi_\mathbf Q) H \rangle_{\psi_\mathbf Q^\prime} =  \langle V_\perp(1-\Pi_\mathbf Q) (H_\mathbf Q + V_\parallel + V_\perp) \rangle_{\psi_\mathbf Q^\prime}  =  \langle V_\perp^2\rangle_{\psi^\perp_\mathbf Q}\;.
\end{equation}
Hence, the varational energy difference is given as:
\begin{equation}
    \Delta E(\lambda)=\langle H \rangle_{\psi^\prime_\mathbf Q}-E(\lambda) =\frac{\lambda}{1+ \lambda^2 \langle V_\perp^2\rangle_{\psi^\prime_\mathbf Q}} \left[ 2 \langle V_\perp^2 \rangle_{\psi_\mathbf Q^\prime}  - {\lambda} (\langle V_\perp H V_\perp\rangle_{\psi^\prime_\mathbf Q}-\langle V_\perp^2\rangle_{\psi^\prime_\mathbf Q}\langle H\rangle_{\psi^\prime_\mathbf Q})\right]\;.
\end{equation}
Since Eq. (\ref{Eq:variationalstabilityproof}) is estabilished, we know that the coefficient of the second term is strictly positive. Thus, for:
\begin{equation}
    0 <\lambda <  \frac{2\langle V_\perp^2 \rangle_{\psi^\prime_\mathbf Q}}{\langle V_\perp H V_\perp \rangle_{\psi^\prime_\mathbf Q}- \langle V^2_\perp\rangle \langle H\rangle_{\psi^\prime_\mathbf Q}}\;,
    \label{eq:positiveinterval}
\end{equation}
we have that $\langle H \rangle_{\psi^\prime_\mathbf Q}- E(\lambda) >0$. Notice that $\Delta E(\lambda) = \langle H \rangle_{\psi^\prime_\mathbf Q}-E(\lambda)$ is a rational function of $\lambda$, and since $1+ \lambda^2\langle V^2_\perp\rangle_{\psi^\prime_\mathbf Q} >0$, it follows that it is also smooth. Its first derivative is given by:
\begin{equation}
    \frac{d}{d\lambda}\Delta E(\lambda) =\frac{2v-2a \lambda - 2v^2\lambda^2}{(1+\lambda^2\langle V_\perp^2\rangle_{\psi^\prime_\mathbf Q})^2}\;,
\end{equation}
where $a=\langle V_\perp H V_\perp\rangle_{\psi^\prime_\mathbf Q} -\langle V_\perp^2\rangle \langle H \rangle_{\psi^\prime_\mathbf Q} >0$ and $v =\langle V_\perp^2\rangle_{\psi^\prime_\mathbf Q} >0$. Our objective is to derive the conditions for the positivity of the derivative. The discriminant of the numerator is $4a^2+16v^3 >0$, and there are a positive and a negative root, given by $(-a \pm \sqrt{a^2+4v^3})/2v^2$. For any $\lambda$ between the roots, the numerator is positive and so is $d\Delta E(\lambda)/d\lambda$. Hence, we can define:
\begin{equation}
    \lambda_\mathrm{max} \equiv \min\left(\frac{2v}{a},\frac{\sqrt{a^2+4 v^3}-a}{2v^2} \right)\;,
\end{equation}
such that $\Delta E(\lambda)$ is , by Eq. (\ref{eq:positiveinterval}), and monotonically increasing due to the positivity of the derivative, for all $\lambda \in (0, \lambda_\mathrm{max})$. 
\end{proof}
By using the characterization of stabilizer ground state energies of App. \ref{app:stabgroundenergy}, we can show that this implies the stabilizer gap, under an extra assumption:
\begin{lem}
    Let $H=H_\mathbf Q + V$ satisfying both conditions of Lemma \ref{lem:degperturbstabilizer}, and suppose that $\mathbf P(V)\backslash \mathbf Q$ do not contain any commuting element with $\mathbf Q$. Then, $|\psi^\prime_\mathbf Q(\lambda)\rangle $ satisfies
    \begin{equation}
        E_\mathrm{STAB}(H) -\langle H\rangle_{\psi^\prime_\mathbf Q(\lambda)} > 0\;,
    \end{equation}
    implying in $\delta_\mathrm{STAB}(H)> 0$ for such Hamiltonians, for all $0 < \lambda < \lambda_\mathrm{max}$.
    \label{lem:stabgapguarantee}
\end{lem}
\begin{proof}
    Note that it follows if we show that weak perturbations do not change the stabilizer ground energy with respect to $H_\mathbf Q$, since if $E_\mathrm{STAB}(H)=\langle H\rangle_{\psi^\prime_\mathbf Q}$, the family of states in Eq. (\ref{eq:ansatzstate}) for $\lambda \in (0,\lambda_\mathrm{max})$ satisfies: 
\begin{equation}
        E_\mathrm{STAB}(H)-E(\lambda) =\Delta E(\lambda) >0\;.
\end{equation}
Let us show that is indeed the case. By Lemma \ref{lem:stabilizerenergy}, we minimize over closed commuting subsets. Given $\mathbf Q^\prime \in C_\mathrm{max}(H)$, its energy is given as:
\begin{equation}
    \tr(H \tilde{\Pi}_{\langle \mathbf Q^\prime \rangle}) = \tr(H_\mathbf Q \tilde{\Pi}_{\langle \mathbf Q^\prime \rangle})-\sum_{P \in \mathbf Q^\prime \cap \mathbf P(V)}w_P\;,
\end{equation}
where we will refer to $\tilde \Pi_\mathbf S = |\mathbf S|\Pi_\mathbf S/2^n$ for any stabilizer group $\mathbf S$ as the natural uniform quantum state on its codespace. We want to show that $\langle H\rangle_{\psi^\prime_\mathbf Q}$ is of the form above. Define:
\begin{equation}
\mathbf S_\mathbf Q\equiv \mathrm{gen}(\mathbf Q)\cup(\mathbf P(V)\cap\mathbf Q) =\mathbf P(H) \cap \mathbf Q\subseteq \mathbf P(H)\;.
\end{equation}
First, let us show that this subset indeed is an element of $C_\mathrm{max}[\mathbf P(H)]$. First, note that $\langle \mathbf S_\mathbf Q\rangle \subseteq \mathbf Q$, and from that we can show that it is in $C[\mathbf P(H)]$, since (1) $\mathbf Q$ is a stabilizer group and therefore $-1 \notin \langle \mathbf S_\mathbf Q\rangle$ and (2) this implies that $\mathbf P(H)\cap \langle \mathbf S_\mathbf Q\rangle \subseteq \mathbf P(H)\cap \mathbf Q$, but this implies $\mathbf P(H) \cap \langle \mathbf S_\mathbf Q\rangle \subseteq\mathbf S_\mathbf Q$. Since $\mathbf S_\mathbf Q\subseteq \mathbf P(H) \cap \langle \mathbf S_\mathbf Q\rangle $, it follows that $\mathbf S_\mathbf Q = \mathbf P(H)\cap \langle \mathbf S_\mathbf Q\rangle$.

It remains to show that it is also maximal. But this follows from the following decomposition of $\mathbf P(H)$:
\begin{equation}
    \mathbf P(H) = \mathrm{gen}(\mathbf Q)\cup \mathbf P(V) = \underbrace{\mathrm{gen}(\mathbf Q) \cup \mathbf P(V)\cap \mathbf Q}_{=\mathbf S_\mathbf Q} \sqcup \mathbf P(V)\backslash \mathbf Q\;,
\end{equation}
where the last union is disjoint. By assumption, we assume that no Paulis in $\mathbf P(V)\backslash \mathbf Q$ commutes with $\mathbf Q$, thatis, they are logical. Hence, every Pauli in $\mathbf P(V)\backslash \mathbf Q$ does anticommute with at least one Pauli in $\mathbf S_\mathbf Q$, and thus showing that $\mathbf S_\mathbf Q$ does not admit any commmuting extension. It follows that $\mathbf S_\mathbf Q \in \mathrm C_\mathrm{max}[\mathbf P(H)]$. Furthermore, we note that:
\begin{equation}
\tr(H\tilde{\Pi}_{\langle \mathbf S_\mathbf Q\rangle}) = -\sum_{P \in \mathrm{gen}(\mathbf Q)}1-\sum_{P \in \mathbf P(V)\cap \mathbf Q}w_P=\langle H\rangle_{\psi^\prime_\mathbf Q}\;,
\end{equation}
showing that $\langle H\rangle_{\psi^\prime_\mathbf Q}$ is in the feasible set of the minimization in Lemma \ref{lem:stabilizerenergy}. Let us show that all other closed commuting subsets have higher energy. Assume $\mathbf Q^\prime \neq \mathbf S_\mathbf Q \in C_\mathrm{max}[\mathbf P(H)]$. Consider the energy difference:
\begin{align}
    \tr(H \tilde{\Pi}_{\langle \mathbf Q^\prime \rangle}) - \langle H\rangle_{\psi^\prime_\mathbf Q} &= \left[\tr(H_\mathbf Q\tilde{\Pi}_{\langle \mathbf Q^\prime \rangle})-\langle H_\mathbf Q\rangle_{\psi^\prime_\mathbf Q} \right] - \sum_{P \in \mathbf Q^\prime \cap \mathbf P(V)}w_P + \sum_{P \in \mathbf Q \cap \mathbf P(V)}w_P\\
    & = \tr[H_\mathbf Q(\tilde{\Pi}_{\langle \mathbf Q^\prime \rangle}-\tilde{\Pi}_{\langle \mathbf S_\mathbf Q \rangle} )]- \sum_{P \in \mathbf Q^\prime \cap \mathbf P(V)}w_P + \sum_{P \in \mathbf Q \cap \mathbf P(V)}w_P\;.
    \label{eq:energydiffstab}
\end{align}
Note that $\tr(H_\mathbf Q\tilde{\Pi}_{\langle \mathbf Q^\prime \rangle}) \geq \tr(H_\mathbf Q\tilde{\Pi}_{\langle \mathbf S_\mathbf Q \rangle})$, where equality is attained if the code includes, $V_{\langle \mathbf Q^\prime\rangle}\subseteq V_{\langle \mathbf S_\mathbf Q\rangle}$, implying $\langle \mathbf Q^ \prime \rangle \subseteq \langle \mathbf S_\mathbf Q\rangle$. Taking the intersection with $\mathbf P(H)$, this yields $\mathbf Q^\prime \subseteq \mathbf S_\mathbf Q$. But we know that this cannot happen, since $\mathbf Q^\prime \subset \mathbf S_\mathbf Q$ would imply that $\mathbf Q^\prime$ is not maximal, and also $\mathbf Q^\prime \neq \mathbf S_\mathbf Q$. Thus, $\tr(H_\mathbf Q\tilde{\Pi}_{\langle \mathbf Q^\prime \rangle}) > \tr(H_\mathbf Q\tilde{\Pi}_{\langle \mathbf S_\mathbf Q \rangle})$. It follows that the energy difference is lower bounded by the spectral gap of $H_\mathbf Q$, given by $\Delta_{gap}=2\min_{P \in \mathbf P(H)}w_P$, giving  the inequality:
\begin{equation}
   \tr(H \tilde{\Pi}_{\langle \mathbf Q^\prime \rangle}) - \langle H\rangle_{\psi^\prime_\mathbf Q}\geq \Delta_\mathrm{gap} - \sum_{P \in \mathbf Q^\prime \cap \mathbf P(V)}w_P + \sum_{P \in \mathbf Q \cap \mathbf P(V)}w_P\;,
\end{equation}
Now, we note that we can split:
\begin{equation}
    \sum_{P \in \mathbf Q^\prime \cap \mathbf P(V)}w_P = \underbrace{\sum_{P \in \mathbf Q \cap \mathbf P(V) \cap \mathbf Q^\prime}w_P}_{\leq \sum_{P \in \mathbf Q \cap \mathbf P(V)}w_P} + \underbrace{\sum_{P \in \mathbf P(V)\backslash \mathbf Q\cap\mathbf Q^\prime}w_P}_{\leq \sum_{P \in \mathbf P(V)}w_P}\;,
\end{equation}
where in the equality, we split the elements in the intersection with respect with their membership in $\mathbf Q$, and the bounds in the underbraces are obtained by adding terms in the sum. Substituting back in Eq. (\ref{eq:energydiffstab}), we obtain:
\begin{align}
    \tr(H \tilde{\Pi}_{\langle \mathbf Q^\prime \rangle}) - \langle H\rangle_{\psi^\prime_\mathbf Q} &\geq \Delta_\mathrm{gap} - \sum_{P \in \mathbf Q \cap \mathbf P(V)}w_P -  \sum_{P \in  \mathbf P(V)}w_P + \sum_{P \in \mathbf Q \cap \mathbf P(V)}w_P\\
    &\geq \Delta_\mathrm{gap}-  \sum_{P \in  \mathbf P(V)}w_P\;,
\end{align}
which is $>0$ by assumption, as discussed in Eq. (\ref{eq:fluctuationenergydiff}). Hence, $E_\mathrm{STAB}(H) = \langle H\rangle_{\psi^\prime_\mathbf Q}$.
\end{proof}

\section{Proof of Lemma \ref{lem:sumofstabilizers} \label{app:proofsumofstabilizers}}
Define $G_{1, \cdots,\ell}$ be the $\ell-$partite commutativity graph of the generators of the corresponding groups, whose vertices decompose as $V(G_{1,\cdots, \ell}) = \cup_{j=1}^m V_j$ with $V_j=\mathrm{gen}(\mathbf S_j)$ and edges are defined if two generators of different groups do not commute. Define the set of maximal independent subsets as:
\begin{equation}
    I_{\max}(G_{1, \cdots, \ell}) =\{\mathbf Q \subseteq  V(G_{1, \cdots, \ell})|\;\forall i,j \in \mathbf Q: (i,j) \notin E(G_{1,\cdots, \ell}) \;\&\;\mathbf Q\mathrm{\;is\;maximal}\}\;,
\end{equation}
where maximality of $\mathbf Q \in I_{\max}(G_{1, \cdots, \ell}) $ refers to the property that there is no other subset of the vertices $\mathbf Q^\prime$ on which $\mathbf Q \subsetneq \mathbf Q^\prime$ satisfying independence.

We claim that this set is one-to-one with $C_{\max}[\mathbf P(H)]$. Since, in this case we have:
\begin{equation}
    \mathbf P(H) =\bigsqcup_{j=1}^\ell \mathrm{gen}(\mathbf S_j) = V(G_{1,\cdots,\ell})\;,
\end{equation}
the two directions can be analyzed:
\begin{enumerate}
    \item Let $\mathbf Q \in C_{\max}[\mathbf P(H)]$. Then, the corresponding vertex set must be independent, since $\mathbf Q$ is commuting, and also maximal, since maximality as a closed commuting subset does also correspond to maximality as a independent set in the graph $G_{1, \cdots, \ell}$.
    \item Given $\mathbf Q \in I_{\max}(G_{1, \cdots, \ell})$, we know that (1) $-1 \notin \langle \mathbf Q \rangle$, since $\mathbf Q \subseteq \mathrm{gen}(\mathbf S_j)$, and thus $\langle \mathbf Q\rangle \subseteq \mathbf S_j$, and (2) it is also closed, due to independence: Since every stabilizer group $\{\mathbf S_j\}_{j=1}^\ell$ cannot be generated as a product of other ones, notice that, by decomposing $\mathbf Q = \sqcup_{j=1}^\ell \mathbf Q_j$:
    \begin{equation}
        \langle \mathbf Q\rangle \cap \mathbf P(H) =\left \langle \bigsqcup_{j=1}^\ell \mathbf Q_j \right \rangle \cap \mathbf P(H) = \bigsqcup_{j=1}^\ell \langle \mathbf Q_j\rangle\cap \mathbf P(H) = \bigsqcup_{j=1}^\ell  \mathbf Q_j\cap \mathbf P(H) = \bigsqcup_{j=1}^\ell  \mathbf Q_j=\mathbf Q\;.
    \end{equation}
\end{enumerate}
Then,
\begin{equation}
    E_\mathrm{STAB}(H) = \max_{\mathbf Q =\sqcup_{j=1}^\ell \mathbf Q_j \in I_{\max}(G_{1,\cdots, \ell})} \sum_{j=1}^\ell |\mathbf Q_j|w_j\;.
\end{equation}
This corresponds to finding the maximum-weight independent set (MWIS) of the graph $G_{1, \cdots, \ell}$ and the weights $\{w_j\}_{j=1}^\ell$. $\square$

\bibliographystyle{apsrev4-2}
\end{document}